\PassOptionsToPackage{scaled=0.86}{helvet}
\documentclass{lmcs} 
\pdfoutput=1

\usepackage{lastpage}
\lmcsdoi{17}{4}{18}
\lmcsheading{}{\pageref{LastPage}}{}{}%
{Nov.~20,~2020}{Dec.~14,~2021}{}

\usepackage[utf8]{inputenc}


\usepackage[misc]{ifsym}
\usepackage{microtype}
\usepackage{mathtools}
\usepackage{nicefrac}
\usepackage{enumitem}
\usepackage{scrextend}
\usepackage{float}
\usepackage{forest}
\usepackage{booktabs}
\usepackage{multirow}
\usepackage{stmaryrd}
\usepackage{amssymb}
\usepackage{pifont}
\usepackage{xargs}
\PassOptionsToPackage{normalem}{ulem}
\usepackage{ulem}
\usepackage{color}
\usepackage{cite}
\usepackage{eqparbox}
\usepackage{url}

\newcommand\paper{article}

\newcommand\cst[1]{\mathsf{#1}}
\newcommand\nf[2]{{#1}_{\downarrow{\ifthenelse{\equal{#2}{}}{\beta}{#2}}}}
\newcommand\unif{\mathrel{\smash{\stackrel{\lower.1ex\hbox{\ensuremath{\scriptscriptstyle ?}}}{=}}}}

\newcommand\ourpara[1]{\subsection*{#1}}

\def\vthinspace{\kern+0.083333em}
\newcommand\unifarrow{\longrightarrow}
\newcommand\jpunifarrow{\Longrightarrow}
\newcommand\ord{\makeop{ord}}

\definecolor{light-gray}{gray}{0.875}
\newcommand{\selected}[1]{\smash{\setlength{\fboxsep}{.3ex}\colorbox{light-gray}{\ensuremath{\vphantom{('q}{#1}}}}} 

\renewcommand\AA{{\textsf{A}}}
\newcommand\BB{{\textsf{B}}}
\newcommand\NN{{\textsf{N}}}
\newcolumntype{C}[1]{>{\centering\let\newline\\\arraybackslash\hspace{0pt}}m{#1}}

\newcommand{\hotofo}[1]{\lfloor#1\rfloor}
\newcommand{\nametodb}[1]{\ensuremath{\langle #1 \rangle_\textsf{db}}}
\newcommand{\cstdb}[2]{\ensuremath{\textsf{db}_{#1}^{#2}}}
\newcommand{\cstdba}[1]{\ensuremath{\cstdb{#1}{\alpha}}}

\newcommand{\ptarrow}[1]{\ensuremath{\;\Longrightarrow^{#1}\;}}
\newcommand{\ptarrowid}[0]{\ptarrow{\text{id}}}

\newcommand{\xmark}{\ding{55}}%

\newcommand{\rulename}[1]{\textsf{#1}}

\let\oldleft\left
\let\oldright\right
\def\left{\mathopen{}\mathclose\bgroup\oldleft}
\def\right{\aftergroup\egroup\oldright}

\global\long\def\makeop#1{\operatorname{#1}}
\global\long\def\seq#1{\overline{#1} }
\newcommandx\param[1][usedefault, addprefix=\global, 1=\seq x]{\lambda#1.\, }
\global\long\def\set#1#2{\left\{  #1\,\middle|\,#2\right\}  }
\newcommandx\jp[3][usedefault, addprefix=\global, 1=3]{\textnormal{JP #2#1.#3}}

\newcommand{\equi}{\leftrightarrow^*_{\alpha\beta\eta}}

\DeclareFontFamily{OT1}{pzc}{}
\DeclareFontShape{OT1}{pzc}{m}{it}{<-> s * [1.10] pzcmi7t}{}
\DeclareMathAlphabet{\mathcalx}{OT1}{pzc}{m}{it}

\begin{document}

\title{Efficient Full Higher-Order Unification}


\author[P.~Vukmirovi\'c]{Petar Vukmirovi\'c}	
\author[A.~Bentkamp]{Alexander Bentkamp}	
\author[V.~Nummelin]{Visa Nummelin}	
\address{Vrije Universiteit Amsterdam,
De Boelelaan 1111,
1081 HV Amsterdam,
The Netherlands}	
\email{\{p.vukmirovic,a.bentkamp,visa.nummelin\}@vu.nl}  




\begin{abstract}
  \noindent We developed a procedure to enumerate complete sets of higher-order unifiers based on work
  by Jensen and Pietrzykowski. Our procedure removes many redundant unifiers by
  carefully restricting the search space and tightly integrating decision
  procedures for fragments that admit a finite complete set of unifiers. We
  identify a new such fragment and describe a procedure for computing its unifiers.
  Our unification procedure, together with new higher-order term indexing data structures,
  is implemented in the Zipperposition theorem prover.
  Experimental evaluation shows a clear advantage over Jensen and Pietrzykowski's
  procedure.
\end{abstract}

\maketitle

\section{Introduction}%
\label{sec:intro}

\looseness=-1
Unification is concerned with finding a substitution that makes two terms
equal, for some notion of syntactic equality. Since the invention of Robinson's first-order unification
algorithm~\cite{ar-resolution-65}, it has become an indispensable
tool in theorem proving, logic
programming, natural language processing, programming language compilation and other areas of computer science.

Many of these applications are based on higher-order formalisms and require
higher-order unification. Due to its undecidability and explosiveness,
the higher-order unification problem is considered one of the main
obstacles on the road to efficient higher-order tools.

One of the reasons for higher-order unification's explosiveness lies in
\emph{flex-flex pairs}, which consist of two variable-headed terms,
e.g., $F \, X \unif G \, \cst{a}$, where $F$, $G$, and
$X$ are variables and $\cst{a}$ is a constant. Even this seemingly simple
problem has infinitely many incomparable unifiers.
One of the first methods designed to combat this explosion is Huet's
preunification~\cite{gh-unification-75}. Huet noticed that some logical calculi
would remain complete if flex-flex pairs are not eagerly solved but postponed as
constraints. If only flex-flex constraints remain, we know that a unifier must
exist and we do not need to solve them.
Huet's preunification has been used in many reasoning tools including Isabelle~\cite{tn-isabelle-2002}, Leo-III~\cite{ascb-leo3-2018}, and Satallax~\cite{cb-satallax-12}. However, recent developments in higher-order theorem
proving~\cite{ab-lamsup-2019,br-combs-19} require full unification---i.e., enumeration of unifiers even for
flex-flex pairs, which is the focus of this \paper{}.

Jensen and Pietrzykowski's (JP) procedure~\cite{jp-unif-76} is the best known
procedure for this purpose (Section~\ref{sec:background}). Given two terms to
unify, it first identifies a position where the terms disagree.
Then, in parallel branches of the search tree, it applies suitable substitutions, involving a
variable either at the position of disagreement or above, and repeats this process on
the resulting terms until they are equal or trivially nonunifiable.

Building on the JP procedure, we designed a new procedure (Section~\ref{sec:the-unification-procedure}) with the same completeness guarantees (Section~\ref{sec:proof-of-completeness}).
The new procedure addresses many of the issues that are
detrimental to the performance of the JP procedure.
First, the JP procedure does not terminate in many cases of obvious
nonunifiability, e.g., for $X \unif \cst{f} \, X$, where $X$ is a non-functional
variable and $\cst{f}$ is a function constant. This example also shows that
the JP procedure does not generalize Robinson's first-order procedure gracefully. To address
this issue, our procedure detects whether a unification problem belongs to a
fragment for which unification is decidable and finite complete sets of unifiers (CSUs)
exist.
We call algorithms that enumerate elements of the CSU for such fragments
\emph{oracles}. Noteworthy fragments with oracles are first-order terms,
patterns~\cite{tn-patterns-93}, functions-as-constructors~\cite{tl-facunif-2016}, and a new fragment
we present in Section~\ref{sec:solid-oracle}.
The unification procedures of Isabelle and Leo-III check whether the unification
problem belongs to a decidable fragment, but we take this idea a step further by
checking this more efficiently and for every subproblem arising during
unification.

Second, the JP procedure computes many redundant unifiers. Consider the
example $F \, (G \, \cst{a}) \unif F \, \cst{b}$, where it produces, in addition
to the desired unifiers $\{F \mapsto \lambda x. \, H\}$ and $\{G \mapsto \lambda
x. \, \cst{b}\}$, the redundant unifier $\{F \mapsto \lambda x. \, H,\; G \mapsto
\lambda x. \, x\}$.
The design of our procedure avoids computing many redundant unifiers, including
this one. Additionally, as oracles usually return a small CSU,
their integration reduces the number of redundant unifiers.

\looseness=-1
Third, the JP procedure applies more explosive rules than Huet's preunification procedure to
flex-rigid pairs. To gracefully generalize Huet's procedure,
we show that his rules for flex-rigid pairs suffice
to enumerate CSUs
if combined with appropriate rules for flex-flex pairs.

Fourth, the JP procedure repeatedly traverses the parts of the unification
problem that have already been unified. Consider the problem $\cst{f}^{100} \,
(G \, \cst{a}) \unif \cst{f}^{100} \, (H \, \cst{b})$, where the exponents
denote repeated application. It is easy to see that this problem can be reduced
to $G \, \cst{a} \unif H \, \cst{b}$. However, the JP procedure will wastefully
retraverse the common context $ \cst{f}^{100} [\;] $ after applying each new
substitution. Since the JP procedure must apply substitutions to the variables occurring in the
common context above the position of disagreement, it cannot be easily adapted to
eagerly decompose unification pairs. By contrast, our procedure is designed to decompose
the pairs eagerly, never traversing a common context twice.

Last, the JP procedure does not allow to apply substitutions and
$\beta$-reduce lazily.
The rules of simpler procedures (e.g., first-order~\cite{kh-unification-2009} and pattern unification~\cite{tn-patterns-93}) depend only on
the heads of the unification pair.
Thus, to
determine the next step, implementations of these procedures need to
substitute and
$\beta$-reduce only until the heads of the current unification
pair are not mapped by the substitution and are not $\lambda$-abstractions.
Since the JP procedure is not based on the decomposition of unification pairs,
it is unfit for
optimizations of this kind.
We designed our procedure to allow for this optimization.

To more efficiently find terms (in a large term set) that are unifiable with a given query term,
we developed a higher-order extension of fingerprint indexing~\cite{ss-fpindex-12}
(Section~\ref{sec:indexing}).
We implemented our procedure, several oracles,
and the fingerprint index in
the Zipperposition prover (Section~\ref{sec:implementation}). Since a
straightforward implementation of the JP procedure already existed in
Zipperposition, we used it as a baseline to evaluate the performance of our
procedure (Section~\ref{sec:evaluation}). The results show substantial
performance improvements.

This invited article is an extended version of our FSCD-2020 paper~\cite{pv-hounif-2020}.
Most notable extension is the Section~\ref{sec:proof-of-completeness}, which gives
the detailed proof of completeness of our new procedure. In addition,
we give proofs for all the unproved statements from the paper, expand the examples
and provide more detailed explanations.
%

\section{Background}%
\label{sec:background}
Our setting is the simply typed $\lambda$-calculus. Types $\alpha, \beta,
\gamma$ are either base types or functional types $\alpha \rightarrow \beta$. By
convention, when we write $\alpha_1 \rightarrow \cdots \rightarrow \alpha_n
\rightarrow \beta$, we assume $\beta$ to be a base type. Basic terms are free
variables (denoted $F, G, H, \dots$), bound variables ($x,y,z$), and constants
($\cst{f}, \cst{g}, \cst{h}$). Complex terms are applications of one term to
another ($s \, t$) or $\lambda$-abstractions~($\lambda x.\, s$).
Following Nipkow~\cite{tn-patterns-93}, we use these syntactic conventions to distinguish free from bound variables.
Bound variables with no enclosing binder, such as $x$ in $\lambda y.\, x$, are called \emph{loose bound
variables}. We say that a
term without loose bound variables is \emph{closed} and a term without free
variables is \emph{ground}. Iterated
$\lambda$-abstraction $\lambda x_1 \ldots \lambda x_n.\, s$ is abbreviated as $\lambda
\overline{x}_n.\, s$ and iterated application $(s \, t_1) \, \ldots\, t_n$ as
$s \, \overline{t}_n$, where $n \geq 0$. Similarly, we denote a sequence of terms $t_1, \ldots,
t_n$ by $\overline{t}_n$, omitting its length $n \geq 0$ where it can be inferred or is
irrelevant.
Parameters and body for any term $\param[\seq x]s$ are defined
to be $\seq x$ and $s$ respectively, where $s$ is
not a $\lambda$-abstraction.
The \emph{size} of a term is inductively defined as
$\text{size}(F) = 1$; $\text{size}(x) = 1$; $\text{size}(\cst{f}) = 1$;
$\text{size}(s\,t) = \text{size}(s) + \text{size}(t)$;
$\text{size}(\lambda x.\, s) = \text{size}(s) + 1$.

\let\rho\varrho 

\looseness=-1
We assume the standard notions of $\alpha$-, $\beta$-, $\eta$-conversions. A term is in
\emph{head normal form} (\emph{hnf}) if it is of the form $\lambda
\overline{x}.\,a\,\overline{t}$, where $a$ is a free variable, bound variable,
or a constant. In this case, $a$ is called the \emph{head} of the term.
By convention, $a$ and $b$ denote heads.
If $a$ is a free variable, we call it a \emph{flex} head; otherwise, we call it a \emph{rigid} head.
A term is called flex or rigid if its head is flex or rigid, respectively.
By $\nf{s}{\textsf{h}}$ we denote the term obtained from a term $s$
by repeated $\beta$-reduction of
the leftmost outermost redex until it is in hnf.
Unless stated otherwise, we view terms syntactically, as opposed to
$\alpha\beta\eta$-equivalence classes.
We write $s \equi t$
if $s$ and $t$ are $\alpha\beta\eta$-equivalent.
Substitutions ($\sigma,\varrho,\theta$) are functions from free and bound variables to terms;
$\sigma t$ denotes application of $\sigma$ to $t$, which $\alpha$-renames $t$ to avoid variable capture. The composition $\varrho\sigma$
of substitutions is defined by $\left(\varrho\sigma\right)t=\varrho\left(\sigma t\right)$.
A variable $F$ is mapped by $\sigma$ if $ \sigma F \not\equi F$.
Given a substitution $\varrho$, which maps $F$ to $s$, we write $\varrho\setminus\{F \mapsto s\}$
to denote a substitution that does not map $F$ and otherwise coincides with $\varrho$.
Given substitutions $\varrho$ and $\sigma$, which map disjoint sets of variables,
we write $\varrho \cup \sigma$ to denote $\varrho\sigma$.

Deviating from the standard notion of higher-order subterm,
we define subterms on $\beta$-reduced terms as follows: a term $t$ is a
subterm of $t$ at position $\varepsilon$. If $s$ is a subterm of
$u_i$ at position $p$, then $s$ is a subterm of $a\;\overline{u}_n$ at
position $i.p$.
If $s$ is a subterm of $t$ at position $p$, then $s$ is a subterm
of $\lambda x.\, t$ at position $1.p$.
Our definition of subterm gracefully generalizes
the corresponding first-order notion: $\cst{a}$ is a
subterm of $\cst{f} \, \cst{a} \, \cst{b}$ at position 1,
but $\cst{f}$ and $\cst{f} \, \cst{a}$ are not subterms
of $\cst{f} \, \cst{a} \, \cst{b}$.
A context is a term with zero or more subterms replaced by a hole $\square$.
We write $C[\overline{u}_n]$ for the term resulting from filling in the holes of
a context $C$ with the terms $\overline{u}_n$ from left to right.
The common context
$\mathcal{C}(s,t)$ of two $\eta$-long $\beta$-reduced terms $s$~and~$t$
of the same type is defined
inductively as follows, assuming that $a \not= b$: $\mathcal{C}(\lambda x.\, s, \lambda y.\, t) = \lambda x.\,
\mathcal{C}(s,\{y\mapsto x\}t)$; $\mathcal{C}(a\,\overline{s}_m, b\,\overline{t}_n) = \square$;
$\mathcal{C}(a\,\overline{s}_m, a\,\overline{t}_m) = a\,\mathcal{C}(s_1,t_1)\,\ldots\,\mathcal{C}(s_m,t_m)$.

A \emph{unification constraint} $s \unif t$ is an unordered pair of two terms of
the same type. A \emph{unifier} of a multiset of unification constraints $E$ is
a substitution $\sigma$, such that $\sigma s \equi \sigma t$, for all $s \unif t
\in E$. A \emph{complete set of unifiers} (\emph{CSU}) of $E$ is defined as a
set $U$ of $E$'s unifiers along with a set $V$ of \emph{auxiliary
variables} such that no $s \unif t \in E$ contains variables from $V$ and for
every unifier $\rho$ of $E$, there exists a $\sigma \in U$ and a substitution
$\theta$ such that for all $X\not\in V$, $\rho X = \theta\sigma X$. A \emph{most
general unifier} (\emph{MGU}) is a one-element CSU\@. A unifier of terms $s$ and $t$
is a unifier of the singleton multiset $\{ s \unif t \}$.

\begin{rem}
  We use this definition of a CSU because
  JP's definition of a CSU, which we have adopted in our FSCD-2020 paper, is flawed.
  JP's definition does not employ the notion of auxiliary variables,
  but instead requires $\rho X = \theta\sigma X$ for all variables mapped by $\rho$.
  This is problematic because nothing prevents $\rho$ from mapping the auxiliary variables.
  For example, $\sigma = \{F \mapsto \lambda x y.\> G\>y\}$ is supposed
  to be an MGU for $F\>\cst{a}\>\cst{c} \unif F\>\cst{b}\>\cst{c}$.
  But for the unifier $\rho = \{F \mapsto \lambda x y.\> y,\> G \mapsto \lambda x.\> \cst{d}\}$,
  without the notion of auxiliary variables,
  there exists no appropriate substitution $\theta$
  because $\rho G = \theta\sigma G$ requires $\theta G = \lambda x.\> \cst{d}$
  and $\rho F = \theta\sigma F$ requires $\theta G = \lambda x.\> x$.
\end{rem}

\section{The Unification Procedure}%
\label{sec:the-unification-procedure}
\looseness=-1
To unify two terms $s$ and $t$, our procedure builds a tree as follows. The
nodes of the tree have the form $(E,\sigma)$, where $E$ is a multiset of
unification constraints $\{(s_1 \unif t_1), \ldots, (s_n \unif t_n)\}$ and
$\sigma$ is the substitution constructed up to that point.  The root node of the
tree is $(\{s \unif t\}, \makeop{id})$, where $\makeop{id}$ is the identity
substitution. The tree is then constructed applying the transitions listed
below. The leaves of the tree are either failure nodes  $\bot$ or substitutions
$\sigma$. Ignoring failure nodes, the set of all substitutions in the leaves
forms a complete set of unifiers for $s$ and $t$. More generally, our procedure can be used to
unify a multiset $E$ of constraints by making the root of the unification tree
$(E, \makeop{id})$.

The procedure requires an infinite supply of fresh free variables. These fresh
variables must be disjoint from the variables occurring in the initial multiset
$E$. Whenever a transition $(E,\sigma) \unifarrow (E',\sigma')$ is made, all
fresh variables used in $\sigma'$ are removed from the supply and cannot be used
again as fresh variables.

\looseness=-1
The transitions are parametrized by a mapping~$\mathcal{P}$ that assigns a set
of substitutions to a unification pair; this mapping abstracts the
concept of unification rules present in other unification procedures. Moreover,
the transitions are parametrized by a selection function~$S$ mapping a multiset
$E$ of unification constraints to one of those constraints $S(E) \in E$, the
\emph{selected} constraint in $E$. The transitions, defined as follows, are only
applied if the $\selected{\text{grayed}}$ constraint is selected.

\begin{description}[leftmargin=3mm]
\item[\rm\rulename{Succeed}]
    $(\varnothing, \sigma) \unifarrow \sigma$
\smallskip
\item[\rm\rulename{Normalize$_{\alpha\eta}$}]
    $(\{\selected{\lambda \overline{x}_m. \, s \unif \lambda \overline{y}_n. \, t}\}\uplus E, \sigma)
    \unifarrow
    (\{\lambda \overline{x}_m. \, s \unif \lambda \overline{x}_m. \, t'\, x_{n+1} \ldots x_m\}\uplus E, \sigma)$\\
    where $m \geq n$, $\overline{x}_m \not= \overline{y}_n$, and $t'=\{y_1\mapsto x_1,\ldots,y_n\mapsto x_n\}t$
\smallskip
\item[\rm\rulename{Normalize$_\beta$}]
    $(\{\selected{\lambda \overline{x}. \, s \unif \lambda \overline{x}. \, t}\}\uplus E, \sigma)
    \unifarrow
    (\{\lambda \overline{x}. \, \nf{s}{\textsf{h}} \unif \lambda \overline{x}. \, \nf{t}{\textsf{h}}\}\uplus E, \sigma)$\\
    where $s$ or $t$ is not in hnf
    \smallskip
\item[\rm\rulename{Dereference}]
    $(\{\selected{\lambda \overline{x}. \, F \, \overline{s} \unif \lambda \overline{x}. \, t}\}\uplus E, \sigma)
    \unifarrow
    (\{\lambda \overline{x}. \, (\sigma F) \,\overline{s} \unif \lambda \overline{x}. \, t\}\uplus E, \sigma)$\\
    where none of the previous transitions apply and $F$ is mapped by $\sigma$
\smallskip
    \item[\rm\rulename{Fail}]
    $(\{\selected{\lambda \overline{x}. \, a \, \overline{s}_m \unif \lambda \overline{x}. \, b \, \overline{t}_n}\}
    \uplus E, \sigma) \unifarrow \bot$\\
    where none of the previous transitions apply, and
    $a$ and $b$ are different rigid heads
\smallskip
\item[\rm\rulename{Delete}]  $(\{\selected{s \unif s}\}\uplus E, \sigma) \unifarrow (E, \sigma)$\\
    where none of the previous transitions apply
\smallskip
\item[\rm\rulename{OracleSucc}]
$(\{\selected{s \unif t}\}\uplus E, \sigma) \unifarrow (E, \varrho \sigma)$\\
    where none of the previous transitions apply,
    some oracle found a finite CSU $U$ for $\sigma s \unif \sigma t $ using fresh auxiliary variables,
    and $\varrho \in U$; if multiple oracles found a CSU, only one of them is considered
\smallskip
\item[\rm\rulename{OracleFail}]
$(\{\selected{s \unif t}\}\uplus E, \sigma) \unifarrow \bot$\\
where none of the previous transitions apply,
and some oracle determined $\sigma s \unif \sigma t$ has no solutions
\smallskip
\item[\rm\rulename{Decompose}]
    $(\{\selected{\lambda \overline{x}. \, a \, \overline{s}_m \unif \lambda \overline{x}. \, a \, \overline{t}_m}\}
    \uplus E, \sigma)
    \unifarrow
    ( \{s_1 \unif t_1, \ldots ,s_m \unif t_m\} \uplus E,
    \sigma)$\\
    where none of the transitions \rulename{Succeed} to \rulename{OracleFail} apply
\smallskip
\item[\rm\textsf{Bind}]
    $(\{\selected{s \unif t}\}
    \uplus E, \sigma)
    \unifarrow
    (\{s \unif t\}
    \uplus E, \rho \sigma)$\\
    where none of the transitions \textsf{Succeed} to \textsf{OracleFail} apply,
    and $\rho\in\mathcal{P}(s \unif t)$.
\end{description}

\noindent
The transitions are designed so that only \rulename{OracleSucc},
\rulename{Decompose}, and \rulename{Bind} can introduce parallel branches in the
constructed tree. \rulename{OracleSucc} can introduce branches
using different unifiers of the CSU, \rulename{Bind} can introduce branches
using different substitutions in $\mathcal{P}$, and \rulename{Decompose}
can be applied in parallel with \rulename{Bind}.

The form of the rules \rulename{OracleSucc} and \rulename{Bind} is similar: both
extend the current substitution. However,
they are designed following different principles. \rulename{OracleSucc} solves
the selected unification constraint using an efficient algorithm applicable only
to certain classes of terms. On the other hand, \rulename{Bind} is applied to
explore the whole search space for any given constraint. These rules are separated
in two to make \rulename{Bind} applicable only if \rulename{OracleSucc}
(or \rulename{OracleFail}) is not, so that possible solutions (or failures) are
detected early.

Our approach is to apply substitutions and $\alpha\beta\eta$-normalize terms
lazily. In this context, laziness means that the transitions
\rulename{Normalize$_{\alpha\eta}$}, \rulename{Normalize$_{\beta}$}, and
\rulename{Dereference} partially normalize and partially apply the constructed
substitution just enough to ensure that the heads are the ones we would get if
the substitution was fully applied and the term was fully normalized.
Additionally, the transitions that modify the constructed substitution,
\rulename{OracleSucc} and \rulename{Bind}, do not apply that substitution to the
unification pairs directly, but only extend it with a new binding. To support
lazy dereferencing, these rules must maintain the
invariant that all substitutions are idempotent. The invariant is easily
preserved if the substitution $\varrho$ from the definition of
\rulename{OracleSucc} and \rulename{Bind} is itself idempotent and no variable
mapped by $\sigma$ occurs in $\varrho F$, for any variable $F$ mapped by
$\varrho$.

The \textsf{OracleSucc} and \textsf{OracleFail} transitions
invoke oracles, such as pattern unification,
to compute a CSU faster,
produce fewer redundant unifiers, and
discover nonunifiability earlier.
In some cases, addition of oracles lets the procedure terminate more often.

In the literature, oracles are usually stated under the assumption that their
input belongs to the appropriate fragment. To check whether a unification
constraint is inside the fragment, we need to fully apply the substitution and
$\beta$-normalize the constraint.
To avoid these expensive operations and enable efficient oracle integration,
oracles must be redesigned to lazily discover whether the terms
belong to their fragment.
Most oracles contain a decomposition operation which requires only a partial application of
the substitution and only partial $\beta$-normalization. If one of the constraints resulting from decomposition is not in
the fragment, the original problem is not in the fragment.
This allows us to detect that the problem is not in the fragment without fully applying
the substitution and $\beta$-normalizing.


%
%
\looseness=-1
The core of the procedure lies in the \textsf{Bind} step, parameterized
by the mapping~$\mathcal{P}$ that determines which substitutions (called \emph{bindings})
to create. The bindings are defined as follows:
%
\begin{description}[itemsep=1\jot,leftmargin=3mm]
    \item[JP-style projection for $F$] Let $F$ be a free variable of
    type $\alpha_1 \rightarrow \cdots \rightarrow \alpha_n \rightarrow \beta$, where
    some $\alpha_i$ is equal to $\beta$ and $n > 0$. Then the JP-style projection binding is
        \[F \mapsto \lambda \overline{x}_n.\,x_i\]
    \item[Huet-style projection for $F$] Let $F$ be a free variable of type $\alpha_1
    \rightarrow \cdots \rightarrow \alpha_n \rightarrow \beta$, where some $\alpha_i
    = \gamma_1 \rightarrow \cdots \rightarrow \gamma_m \rightarrow \beta$, $n > 0$ and $m \geq 0$. Huet-style projection is
    \[F \mapsto \lambda \overline{x}_n. \, x_i \, (F_1 \, \overline{x}_n) \, \ldots \, (F_m \, \overline{x}_n)\]
    where the fresh free variables $\overline{F}_m$ and bound variables $\overline{x}_n$ are of appropriate types.
    \item[Imitation of $a$ for $F$]
    \looseness=-1 Let $F$ be a free variable of type $\alpha_1 \rightarrow
    \cdots \rightarrow \alpha_n \rightarrow \beta$ and $a$ be a free variable or a constant
    of type $\gamma_1 \rightarrow \cdots
    \rightarrow \gamma_m \rightarrow \beta$ where $n,m \geq 0$. The imitation binding is
    \[F \mapsto \lambda \overline{x}_n. \, a \, (F_1 \, \overline{x}_n) \ldots
    (F_m \, \overline{x}_n)\] where the fresh free variables $\overline{F}_m $ and
    bound variables $\overline{x}_n$ are of appropriate types.
    \item[Elimination for $F$] Let $F$ be a free variable of type $\alpha_1 \rightarrow
    \cdots \rightarrow \alpha_n \rightarrow \beta$, where $n >0$. In addition, let $1 \leq j_1 < \cdots < j_i \leq n$ and $i<n$. Elimination
    for the sequence ${(j_k)}_{k=1}^i$ is
    \[ F \mapsto \lambda \overline{x}_n. \, G \, x_{j_1} \, \ldots \, x_{j_i}\]
    where the fresh free variable $G$ as well as all $x_{j_k}$ are of appropriate type.
    We call fresh variables emerging from this binding in the role of $G$
    \emph{elimination variables}.
    \item[Identification for $F$ and $G$] Let $F$ and $G$ be different free variables. Furthermore, let
    the type of $F$ be $\alpha_1 \rightarrow \cdots \rightarrow \alpha_n
    \rightarrow \beta$ and the type of $G$ be $\gamma_1 \rightarrow \cdots
    \rightarrow \gamma_m \rightarrow \beta$, where $n,m\geq 0$. Then, the identification binding binds
    $F$ and $G$ with
    \begin{equation*}
        F \mapsto \lambda \overline{x}_n. \, H \, \overline{x}_n \, (F_1 \, \overline{x}_n)
        \ldots (F_m \, \overline{x}_n) \quad
        G \mapsto \lambda \overline{y}_m. \,
        H \, (G_1 \, \overline{y}_m) \ldots (G_n \, \overline{y}_m) \, \overline{y}_m
    \end{equation*}
    \looseness=-1
    where the fresh free variables $H,\overline{F}_m,\overline{G}_n$ and bound
    variables $\overline{x}_n$,$\overline{y}_m$ are of appropriate types.
    Fresh variables from this binding with the role of $H$ are called \emph{identification variables}.
    \item[Iteration for $F$]
    \looseness=-1
    Let $F$ be a free variable of the type $\alpha_1 \rightarrow \cdots
    \rightarrow \alpha_n \rightarrow \beta_1$ and let some $\alpha_i$ be the
    type $\gamma_1 \rightarrow \cdots \rightarrow \gamma_m \rightarrow
    \beta_2$, where $n>0$ and $m\geq 0$. Iteration for $F$ at $i$ is
    \[ F \mapsto \lambda \overline{x}_n.\,H\,\overline{x}_n \, (\lambda
    \overline{y}.\,x_i \, (G_1 \, \overline{x}_n \, \overline{y}) \ldots (G_m \,
    \overline{x}_n \, \overline{y}) )\]
    The free variables $H$ and $G_1, \ldots,
    G_m$ are fresh, and $\overline{y}$ is an arbitrary-length sequence of bound
    variables of arbitrary types. All new variables are of
    appropriate type. Due to indeterminacy of $\overline{y}$, this step is
    infinitely branching.
\end{description}

\noindent
The following mapping $\mathcal{P}_\mathsf{c}(\lambda \overline{x}.
\, s \unif \lambda \overline{x}. \, t)$ is used as the parameter $\mathcal{P}$ of the procedure:
\begin{itemize}
  \setlength\itemsep{1\jot}

    \item If the constraint is rigid-rigid, $\mathcal{P}_\mathsf{c}(\lambda \overline{x}.
    \, s \unif \lambda \overline{x}. \, t) = \varnothing$.

    \item If the constraint is flex-rigid,
    let
    $\mathcal{P}_\mathsf{c}(\lambda \overline{x}. \, F \,\overline{s} \unif \lambda \overline{x}. \, a\, \overline{t})$
    be
    \begin{itemize}
        \item an imitation of $a$ for $F$, if $a$ is a constant, and
        \item all Huet-style projections for $F$, if $F$ is not an identification variable.
    \end{itemize}

    \item If the constraint is flex-flex and the heads are different,
    let
    $\mathcal{P}_\mathsf{c}(\lambda \overline{x}. \, F \,\overline{s} \unif \lambda \overline{x}. \, G\, \overline{t})$
    be
    \begin{itemize}
        \item all identifications and iterations for both $F$ and $G$, and
        \item all JP-style projections for non-identification variables among $F$ and $G$.
    \end{itemize}

    \item If the constraint is flex-flex and the heads are identical,
    we distinguish two cases:
    \begin{itemize}
        \item if the head is an elimination variable,
        $\mathcal{P}_\mathsf{c}(\lambda \overline{x}.\, s \unif \lambda \overline{x}.\, t ) = \varnothing$;
        \item otherwise, let
        $\mathcal{P}_\mathsf{c}(\lambda \overline{x}. \, F \,\overline{s} \unif \lambda \overline{x}. \, F\, \overline{t})$
        be all iterations for $F$ at arguments of functional type and
        all eliminations for~$F$.
    \end{itemize}
\end{itemize}
\ourpara{Comparison with the JP Procedure}
%
%
\looseness=-1
The JP procedure enumerates unifiers by constructing a search tree with nodes of the form $(s \unif t, \sigma)$, where $s \unif t$ is the current unification
problem and $\sigma$ is the substitution built so far.
The initial node consists of the input problem and the
identity substitution. Success nodes are nodes of the form $(s \unif s, \sigma)$.
The set of all substitutions contained success nodes form a CSU\@.

To determine the
child nodes of a node $(s \unif t, \sigma)$, the procedure  computes the common context $C$ of $s$ and $t$,
yielding term pairs $(s_1, t_1), \ldots, (s_n, t_n)$, called \emph{disagreement
pairs}, such that $s = C[s_1,\ldots,s_n]$ and $t = C[t_1,\ldots,t_n]$. It
chooses one of the disagreement pairs $(s_i, t_i)$.
Depending on the context $C$ and the chosen disagreement pair $(s_i, t_i)$,
it determines a set of bindings $\mathcal{P}_\mathsf{JP}(C,s_i, t_i)$.
For each of the bindings $\rho\in\mathcal{P}_\mathsf{JP}(C,s_i, t_i)$, it creates a child node
$(\nf{(\varrho{s})}{\beta\eta} \unif \nf{(\varrho{t})}{\beta\eta}, \varrho\sigma)$,
where $\nf{u}{\beta\eta}$ denotes a
$\beta\eta$-normal form of a term $u$.

The set of bindings $\mathcal{P}_\mathsf{JP}(C,s_i, t_i)$
is based on the heads of $s_i$ and $t_i$, and the free variables occurring above
$s_i$ and $t_i$ in $C$. The set $\mathcal{P}_\mathsf{JP}(C,s_i, t_i)$ contains
\begin{itemize}
  \item all JP-style projections for free variables that are heads of $s_i$ or $t_i$;%
\footnote{In JP's formulation of projection, they explicitly mention that the projected argument must be of base type.
In our presentation, this follows from $\beta$ being of base type by the convention introduced in Section~\ref{sec:background}.}

  \item an imitation of $a$ for $F$ if a free variable $F$ is the head of $s_i$ and a free
  variable or constant $a$ is the head of $t_i$ (or vice versa);
  \item all eliminations for free variables occurring above the
  chosen disagreement pair
  eliminating only the argument containing the disagreement pair;
  \item an identification for the heads of $s_i$ and $t_i$ if they are both free variables; and
  \item all iterations for the heads of $s_i$ and
  $t_i$ if they are free variables, and for all free variables occurring above the disagreement pair.%
\footnote{In JP's formulation of iteration, it is not immediately obvious whether they intend to require iteration of arguments of base type.
However, their Definition 2.4~\cite{jp-unif-76} shows that they do.}
\end{itemize}

Architecturally, the most noticeable difference between the JP procedure and ours is
the representation of the problem: The JP procedure works on a single constraint,
while our procedure maintains a multiset of constraints. At a first glance, this
is a merely presentational change. However, it has consequences for termination,
performance, and redundancy of the procedure.

\looseness=-1
Since the JP procedure never decomposes the common context of its only constraint, it allows
iteration or elimination to be applied at a free variable above the disagreement
pair, even if  bindings were already applied below that free variable. This can
lead to many different paths to the same unifier. In contrast, our procedure
makes the decision which binding to apply to a flex-flex pair with the same
head as soon as it is observed. Also, it explores the possibility of not
applying a binding and decomposing the pair. In either way, the flex-flex pair
is never revisited, which improves the performance and returns fewer redundant
unifiers. We show that this restriction prunes the search space without
influencing the completeness.

Our procedure makes the choice of child nodes based only on the heads of the
chosen unification constraint. In contrast, the JP procedure tracks all the
variables occurring in the common context. Thus, lazy normalization and lazy
variable substitution cannot be integrated in the JP procedure a straightforward
fashion. Moreover, as it does not feature a rule similar to
\rulename{Decompose}, it always retraverses the already unified part of the
problem, resulting in poor performance on deep terms.

\looseness=-1
One of the main drawbacks of the JP procedure is that it features a highly
explosive, infinitely branching iteration rule. This rule is a more general
version of Huet-style projection. Its universality enables finding elements of
CSU for flex-flex pairs, for which Huet-style projection does not suffice.
However, the JP procedure applies iteration indiscriminately on both flex-flex and
flex-rigid pairs. We discovered that our procedure remains complete if iteration
is applied only on flex-flex pairs, and Huet-style projection only on
flex-rigid ones. This helps our procedure terminate more often than the JP
procedure. As a side-effect, the restriction of our procedure to
the preunification problem is a graceful generalization of Huet procedure, with
additional improvements such as oracles, lazy substitution, and lazy
$\beta$-reduction.

The bindings of our procedure contain further optimizations that are absent in
the JP procedure: The JP procedure applies eliminations for only one parameter
at a time, yielding multiple paths to the same unifier. It applies imitations to
flex-flex pairs, which we found to be unnecessary. Similarly, we found out that
tracking which rules introduced which variables can avoid computing redundant
unifiers: It is not necessary to apply iterations and eliminations on
elimination variables, and projections on identification variables.

\ourpara{Examples}
We present some examples that demonstrate advantages of our procedure.
The displayed branches of the constructed trees are not necessarily exhaustive.
We abbreviate
JP-style projection as \textsf{JP\,Proj},
imitation as \textsf{Imit},
identification as \textsf{Id},
\rulename{Decompose} as \textsf{Dc},
\rulename{Dereference} as \textsf{Dr},
$\textsf{Normalize}_\beta$ as $\textsf{N}_\beta$,
and
\textsf{Bind} of a binding $x$ as $\textsf{B(}x\textsf{)}$.
Transitions of the JP procedure are denoted by $\Longrightarrow$.
For the JP transitions we implicitly apply the generated bindings and fully
normalize terms, which significantly shortens JP derivations.

\begin{exa}
The JP procedure does not terminate on the problem $G \unif \cst{f} \, G$:
\[(G \unif \cst{f} \, G, \makeop{id})
\overset{\textsf{Imit}}{\jpunifarrow}
(\cst{f}\,G' \unif \cst{f}^2 \, G', \sigma_1)
\overset{\textsf{Imit}}{\jpunifarrow}
(\cst{f}^2 \, G'' \unif \cst{f}^3 \, G'', \sigma_2)
\overset{\textsf{Imit}}{\jpunifarrow}\cdots\]
where $\sigma_1 = \{G \mapsto  \lambda x. \, \cst{f} \, G' \}$
and $\sigma_2 = \{G' \mapsto  \lambda x. \, \cst{f} \, G'' \}\sigma_1$.
By including any oracle that supports the first-order occurs check, such as the
pattern oracle or the fixpoint oracle described in Section~\ref{sec:implementation},
our procedure gracefully generalizes first-order unification:
\[(\{G \unif \cst{f} \, G\}, \makeop{id})
\overset{\textsf{OracleFail}}{\unifarrow}
\bot
\]
\end{exa}

\begin{exa}
The following derivation illustrates the advantage of the \textsf{Decompose} rule.
\begin{align*}
&
(\{\cst{h}^{100} \, (F \, \cst{a}) \unif \cst{h}^{100} \, (G \,\cst{b})\}, \makeop{id})
\overset{\textsf{Dc}^{100}}{\unifarrow}
(\{F \, \cst{a} \unif G \,\cst{b}\}, \makeop{id})
\overset{\textsf{B(Id)}}{\unifarrow}
(\{F \, \cst{a} \unif G \,\cst{b}\}, \sigma_1)
\\ &
\overset{\textsf{Dr}+\textsf{N}_\beta}{\unifarrow}
(\{H \, \cst{a}\, (F' \, \cst{a}) \unif H \, (G'\,\cst{b}) \,\cst{b}\}, \sigma_1)
\overset{\textsf{Dc}}{\unifarrow}
(\{\cst{a} \unif G'\,\cst{b}, F' \, \cst{a} \unif \cst{b}\}, \sigma_1)
\\ &
\overset{\textsf{B(Imit)}}{\unifarrow}
(\{\cst{a} \unif G'\,\cst{b}, F' \, \cst{a} \unif \cst{b}\}, \sigma_2)
\overset{\textsf{Dr}+\textsf{N}_\beta}{\unifarrow}
(\{\cst{a} \unif \cst{a}, F' \, \cst{a} \unif \cst{b}\}, \sigma_2)
\overset{\textsf{Delete}}{\unifarrow}
(\{F' \, \cst{a} \unif \cst{b}\}, \sigma_2) \\
&\overset{\textsf{B(Imit)}}{\unifarrow}
(\{F' \, \cst{a} \unif \cst{b}\}, \sigma_3)
\overset{\textsf{Dr}+\textsf{N}_\beta}{\unifarrow}
(\{\cst{b} \unif \cst{b}\}, \sigma_3)
\overset{\textsf{Delete}}{\unifarrow}
(\varnothing, \sigma_3)
\overset{\textsf{Succeed}}{\unifarrow}
\sigma_3
\end{align*}
where
$\sigma_1 = \{
    F\mapsto \lambda x. \, H \, x\, (F' \, x),
    G\mapsto \lambda y. \, H \, (G' \, y) \, y
    \}$;
$\sigma_2 = \{
  G' \mapsto \lambda x. \, \cst{a}
  \}\sigma_1$; and
$\sigma_3 = \{
  F' \mapsto \lambda x. \, \cst{b}
  \}\sigma_2$.
The JP procedure produces the same intermediate substitutions $\sigma_1$ to $\sigma_3$,
but since it does not decompose the terms, it retraverses the common context
$\cst{h}^{100}\,[\;]$ at every step to identify the contained disagreement pair:
\pagebreak[1]
\begin{align*}
    &
    (\cst{h}^{100} \, (F \, \cst{a}) \unif \cst{h}^{100} \, (G \,\cst{b}), \makeop{id})
    \overset{\textsf{Id}}{\jpunifarrow}
    (\cst{h}^{100} \, (H \, \cst{a}\, (F' \, \cst{a})) \unif \cst{h}^{100} \, (H \, (G' \,\cst{b}) \,\cst{b}), \sigma_1)
    \\ &
    \overset{\textsf{Imit}}{\jpunifarrow}
    (\cst{h}^{100} \, (H \, \cst{a}\, (F' \, \cst{a})) \unif \cst{h}^{100} \, (H \, \cst{a} \,\cst{b}), \sigma_2)
    \overset{\textsf{Imit}}{\jpunifarrow}
    (\cst{h}^{100} \, (H \, \cst{a}\, \cst{b}) \unif \cst{h}^{100} \, (H \, \cst{a} \,\cst{b}), \sigma_3)
    \overset{\textsf{Succeed}}{\jpunifarrow}
    \sigma_3
\end{align*}

\end{exa}

\begin{exa}
  Even when no oracles are used, our procedure performs better than the JP procedure on small, simple
  problems. Consider the problem $F \, \cst{a} \unif \cst{a}$, which has a two
  element CSU:\@ $\{ F \mapsto \lambda x. \, x, F \mapsto \lambda x. \, \cst{a}
  \}$. Our procedure terminates, finding both unifiers:

\begin{align*}
    &
    (\{ F \, \cst{a} \unif \cst{a} \}, \makeop{id})
    \overset{\textsf{B(JP Proj)}}{\unifarrow}
    (\{ F \, \cst{a} \unif \cst{a} \},  \{ F \mapsto \lambda x. \, x \} )
    \overset{\textsf{Dr}+\textsf{N}_\beta}{\unifarrow}
    (\{ \cst{a} \unif \cst{a} \}, \{  F \mapsto \lambda x. \, x \})
    \\ &
    \overset{\textsf{Delete}}{\unifarrow}
    (\varnothing, \{  F \mapsto \lambda x. \, x \})
    \overset{\textsf{Succeed}}{\unifarrow}
    \{  F \mapsto \lambda x. \, x \}
    \\[1\jot]
    &
    (\{ F \, \cst{a} \unif \cst{a} \}, \makeop{id})
    \overset{\textsf{B(Imit)}}{\unifarrow}
    (\{ F \, \cst{a} \unif \cst{a} \},  \{ F \mapsto \lambda x. \, \cst{a} \} )
    \overset{\textsf{Dr}+\textsf{N}_\beta}{\unifarrow}
    (\{ \cst{a} \unif \cst{a} \}, \{  F \mapsto \lambda x. \, \cst{a} \})
    \\ &
    \overset{\textsf{Delete}}{\unifarrow}
    (\varnothing, \{  F \mapsto \lambda x. \, \cst{a} \})
    \overset{\textsf{Succeed}}{\unifarrow}
    \{  F \mapsto \lambda x. \, \cst{a} \}
\end{align*}

The JP procedure finds those two unifiers as well, but it does not terminate as
it applies iterations to $F$.

\end{exa}

\begin{exa}
The search space restrictions also allow us to prune some redundant unifiers. Consider the problem $F \, (G\,\cst{a}) \unif F \,
\cst{b}$, where $\cst{a}$ and $\cst{b}$ are of base type.
Our procedure produces only one failing branch and the following two successful branches:
\begin{align*}
&
(\{ F\,(G\,\cst{a}) \unif F \, \cst{b} \}, \makeop{id})
\overset{\textsf{Dc}}{\unifarrow}
(\{ G\,\cst{a} \unif \cst{b} \}, \makeop{id})
\overset{\textsf{B(Imit)}}{\unifarrow}
(\{ G\,\cst{a} \unif \cst{b} \}, \{ G \mapsto \lambda x.\, \cst{b} \})
\\ &
\overset{\textsf{Dr}+\textsf{N}_\beta}{\unifarrow}
(\{ \cst{b} \unif \cst{b} \}, \{ G \mapsto \lambda x.\, \cst{b} \})
\overset{\textsf{Delete}}{\unifarrow}
(\varnothing, \{ G \mapsto \lambda x.\, \cst{b} \})
\overset{\textsf{Succeed}}{\unifarrow}
\{ G \mapsto \lambda x.\, \cst{b} \}
\\[1\jot]
&
(\{ F\,(G\,\cst{a}) \unif F \, \cst{b} \}, \makeop{id})
\overset{\textsf{B(Elim)}}{\unifarrow}
(\{ F\,(G\,\cst{a}) \unif F \, \cst{b} \}, \{ F \mapsto \lambda x.\, F' \})
\\ &
\overset{\textsf{Dr}+\textsf{N}_\beta}{\unifarrow}
(\{ F' \unif F' \}, \{ F \mapsto \lambda x.\, F' \})
\overset{\textsf{Delete}}{\unifarrow}
(\varnothing, \{ F \mapsto \lambda x.\, F' \})
\overset{\textsf{Succeed}}{\unifarrow}
\{ F \mapsto \lambda x.\, F' \}
\end{align*}
The JP procedure additionally produces the following redundant unifier:
\begin{align*}
    &
    ( F\,(G\,\cst{a}) \unif F \, \cst{b} , \makeop{id})
    \overset{\textsf{JP Proj}}{\jpunifarrow}
    ( F \, \cst{a} = F \, \cst{b} , \{ G \mapsto \lambda x. \, x \})
    \\ &
    \overset{\textsf{Elim}}{\jpunifarrow}
    (F' = F' , \{ G \mapsto \lambda x. \, x, F \mapsto \lambda x. \, F' \})
    \overset{\textsf{Succeed}}{\jpunifarrow}
    \{ G \mapsto \lambda x. \, x, F \mapsto \lambda x. \, F' \}
\end{align*}
Moreover, the JP procedure does not terminate because an infinite number of iterations is
applicable at the root. Our procedure terminates in this case since we only apply iteration
binding for non base-type arguments, which $F$ does not have.
\end{exa}


\ourpara{Pragmatic Variant} We structured our procedure so that most of
the unification machinery is contained in the \rulename{Bind} step. Modifying
$\mathcal{P}$, we can sacrifice completeness and obtain a pragmatic variant of
the procedure that often performs better in practice.
Our preliminary experiments showed that using mapping $\mathcal{P}_\mathsf{p}$ defined as follows is a reasonable compromise between completeness and
performance:
\begin{itemize}
    \setlength\itemsep{1\jot}
    \item If the constraint is rigid-rigid, $\mathcal{P}_\mathsf{p}(\lambda \overline{x}.
    \, s \unif \lambda \overline{x}. \, t) = \varnothing$.
    \item If the constraint is flex-rigid,
    let
    $\mathcal{P}_\mathsf{p}(\lambda \overline{x}. \, F \,\overline{s} \unif \lambda \overline{x}. \, a\, \overline{t})$
    be
    \begin{itemize}
        \item an imitation of $a$ for $F$, if $a$ is a constant, and
        \item all Huet-style projections for $F$ if $F$ is not an identification variable.
    \end{itemize}

    \item If the constraint is flex-flex and the heads are different,
    let
    $\mathcal{P}_\mathsf{p}(\lambda \overline{x}. \, F \,\overline{s} \unif \lambda \overline{x}. \, G\, \overline{t})$
    be
    \begin{itemize}
      \item an identification binding for $F$ and $G$, and
      \item all Huet-style projections for $F$ if $F$ is not an identification variable
  \end{itemize}

  \item If the constraint is flex-flex and the heads are identical, we distinguish two cases:
      \begin{itemize}
      \item if the head is an elimination variable,
      $\mathcal{P}_\mathsf{p}(\lambda \overline{x}. \, F \,\overline{s} \unif \lambda \overline{x}. \, F\, \overline{t}) = \varnothing$;
      \item otherwise, let $\mathcal{P}_\mathsf{p}(\lambda \overline{x}. \, F \,\overline{s} \unif \lambda \overline{x}. \, F\, \overline{t})$
            be the set of all eliminations bindings for $F$.
  \end{itemize}
\end{itemize}

The pragmatic variant of our procedure removes all iteration bindings to enforce finite branching.
Moreover, it
imposes limits
on the number of bindings applied,
counting the applications of bindings locally, per
constraint.
It is useful to distinguish the Huet-style projection cases where
$\alpha_i$ is a base type (called \emph{simple projection}), which always
reduces the problem size, and the cases where $\alpha_i$ is a functional type (called \emph{functional
projection}). We limit the number applications of the following bindings:
functional projections, eliminations, imitations and identifications.
In addition,
a limit on the total number of applied bindings can be set.
An elimination binding that removes $k$ arguments counts as $k$
elimination steps.
Due to these limits, the pragmatic variant
terminates.

To fail as soon as any of the limits
is reached, the pragmatic variant employs an additional oracle.
If this oracle determines that the limits are reached and the constraint is of the form
$\lambda \overline{x}.\, F \, \overline{s}_m \unif \lambda \overline{x}.\, G \, \overline{t}_n$, it returns a \emph{trivial unifier}---a substitution $\{ F \mapsto \lambda \overline{x}_m.\, H, G \mapsto \lambda \overline{x}_n.\, H \}$, where $H$ is a fresh variable; if the limits are reached and
the constraint is flex-rigid, the oracle fails; if the limits are not reached,
it reports that terms are outside its fragment. The trivial unifier prevents the procedure from failing on
easily unifiable flex-flex pairs.

Careful tuning of each limit optimizes the procedure for a specific class of problems.
For problems originating from proof assistants, shallow unification depth usually suffices. However, hard
hand-crafted problems often need
deeper unification.

    \section{Proof of Completeness}%
    \label{sec:proof-of-completeness}


  Like the JP procedure, our procedure misses no unifiers:

  \begin{thm}%
  \label{thm:completeness}
  The procedure described in Section~\ref{sec:the-unification-procedure}
  parametrized by $\mathcal{P}_\mathsf{c}$ is complete, meaning that the
  substitutions on the leaves of the constructed tree form a CSU\@.
  More precisely, let $E$ be a multiset of constraints
  and let $V$ be the supply of fresh variables provided to the procedure.
  Then for any unifier $\varrho$ of $E$ there exists
  a derivation $\left(E,\makeop{id}\right)\longrightarrow^{*}\sigma$ and a
  substitution $\theta$ such that for all free variables $X\not\in V$, we have $\varrho X = \theta\sigma X$.
  \end{thm}

  Taking a high-level view, this theorem is proved by incrementally defining states
  $\left(E_{j},\sigma_{j}\right)$ and \emph{remainder} \emph{substitutions}
  $\varrho_{j}$ starting with
  $\left(E_{0},\sigma_{0}\right)=\left(E,\makeop{id}\right)$ and
  $\varrho_{0}=\varrho$. The substitution $\varrho_{j}$ is what remains to be added
  to $\sigma_{j}$ to reach $\varrho_{0}$. States are defined
  so that the shape of the selected constraint from $E_{j}$ and the remainder
  substitution guide the choice of applicable transition rule.
  We employ a measure based on values of $E_j$ and $\varrho_j$ that decreases with
  each application of the rules. Therefore, eventually, we will reach the target substitution
  $\sigma$.

    \looseness=-1
    In the remaining of this section, we view terms as
    $\alpha\beta\eta$-equivalence classes, with the $\eta$-long $\beta$-normal
    form as their canonical representative. Moreover, we consider all
    substitutions to be fully applied. These assumptions are justified because
    all bindings depend only on the head of terms and hence replacing the lazy
    transitions $\rulename{Normalize}_{\alpha\eta}$,
    $\rulename{Normalize}_{\beta}$, and $\rulename{Dereference}$ by eager
    counterparts only affects the efficiency but not the overall behavior of our
    procedure.

    We now give the detailed completeness proof of Theorem~\ref{thm:completeness}.
    Our proof is an adaptation of the proof given by Jensen and Pietrzykowski~\cite{jp-unif-76}.
    Definitions and lemmas are reused, but are combined together differently to
    suit our procedure. We start by listing all reused definitions and lemmas
    from the original JP proof. The ``JP'' labels in
    their statements refer to the corresponding lemmas and definitions from
    the original proof.


    \begin{defi}[{$\jp[1]D6$}]\label{def:disagreement-pairs}
    \looseness=-1
    Given two terms $t$ and $s$ and their common context $C$,
    we can write $t$ as $C[\seq t]$ and $s$ as $C[\seq s]$
    for some $\seq t$ and $\seq s$.
    The pairs $(s_{j},t_{j})$ are called \emph{disagreement pairs.}
    \end{defi}
    \begin{defi}
    [$\jp D1$]\label{def:opponents}
    Given two terms $t$ and $s$,
    let $\param[\seq x] t'$ and $\param[\seq y] s'$
    be respective $\alpha$-equivalent terms such that
    their parameters $\seq x$ and $\seq y$ are disjoint.
    Then the disagreement pairs of $t'$ and $s'$ are called
    \emph{opponent pairs} in $t$ and $s$.
    \end{defi}
    \begin{lem}
      [$\jp L{3\text{ (1)}}$]\label{lem:principal-lemma}Let $\varrho$ be a substitution
      and $X$, $Y$ be free variables such that $\varrho\left(X\,\seq s\right)=\varrho\left(Y\,\seq t\right)$
      for some term tuples $\seq s$ and $\seq t$. Then for every opponent
      pair $u$, $v$ in $\varrho X$ and $\varrho Y$
      (Definition~\ref{def:opponents}), the head of $u$ or $v$ is a parameter
      of $\varrho X$ or $\varrho Y$.
      \end{lem}

    In contrast to applied constants, applied variables should not be eagerly decomposed.
    For a constant $\cst{f}$, if $\cst{f}\,\seq s \unif \cst{f}\,\seq t$ has a unifier,
    that unifier must clearly also unify $s_{i} \unif t_{i}$ for each~$i$.
    For a free variable $X$, a unifier of $X\,\seq s \unif X\,\seq t$ does not necessarily
    unify $s_{i} \unif t_{i}$.
    The concept of $\omega$-simplicity is
    a criterion on unifiers that captures some of the cases where eager decomposition is possible.
    Non-$\omega$-simplicity
    on the other hand is the main trigger of $\rulename{Iteration}$\textemdash the
    most explosive binding of our procedure.

    \begin{defi}
    [$\jp D2$]\label{def:omega-simple}An occurrence of a parameter $x$ of term $t$
    in the body of $t$ is $\omega$\emph{-simple }if both
    \begin{enumerate}
    \item the arguments of $x$ are distinct and are exactly (the $\eta$-long
    forms of) all of the variables bound in the body of $t$, and
    \item this occurrence of $x$ is not in an argument of any parameter
    of $t$.
    \end{enumerate}
    \end{defi}
    This definition is slightly too restrictive for our purposes.
    It is unfortunate that condition~1 requires $x$ to be applied to
    \emph{all} instead of just some of the bound variables. The JP
    proof would probably work with such a relaxation, and the definition
    would then cover all cases where eager decomposition is possible.
    However, to reuse the JP lemmas, we stick to the original
    notion of $\omega$-simplicity and introduce the following relaxation:
    \begin{defi}
    An occurrence of a parameter $x$ of term $t$ in the body of $t$
    is \emph{base-simple} if it is $\omega$-simple or both
    \begin{enumerate}
    \item $x$ is of base type, and
    \item this occurrence of $x$ is not in an argument of any parameter
    of $t$.
    \end{enumerate}
    \end{defi}
    \begin{lem}\label{lem:base-simple->inj.}
    Let $s$ have parameters $\seq x$ and a subterm $x_{j}\,\seq v$
    where this occurrence of $x_{j}$ is base-simple.
    Then for any sequence $\seq t$ of (at least $j$) terms, the body of
    $t_{j}$ is a subterm of $s\,\seq t$ (after normalization) at the position
    of $x_{j}\,\seq v$ up to renaming of the parameters of $t_{j}$.
    To compare positions of $s$ and $s\,\seq t$, ignore
    the parameter count mismatch.
    \end{lem}
    \begin{proof}
    Consider the process of $\beta$-normalizing $s\,\seq t$. After substituting
    terms $\seq t$ into the body of~$s$, a further reduction can only
    take place when some $t_{k}$ is an abstraction that gets arguments in
    $s$. The arguments $\seq v$ to the $x_{j}$ are distinct variables
    bound in the body of $s$. This follows easily from either case of the
    definition of base-simplicity. So $t_{j}$ is applied to the unmodified
    $\seq v$ after substituting
    terms $\seq t$ into the body of~$s$. Base-simplicity also implies
    that $t_{j}\,\seq v$ does not occur in an argument to another
    $t_{k}$. Hence only the reduction of $t_{j}\,\seq v$ itself affects this
    subterm. The variables $\seq v$ match the parameter count of $t_{j}$
    because we consider the $\eta$-long form of $t_{j}$; so $t_{j}\,\seq v$ reduces to the body
    of $t_{j}$ (modulo renaming). The position is obviously that of $x_{j}\,\seq v$.
    \end{proof}
    \begin{lem}
      [$\jp C{4 strengthened}$]\label{lem:eq-if-omega-simple}Let $\varrho$ be a substitution and $X$ a free
      variable. If $\varrho\left(\param X\,\seq s\right)=\varrho\left(\param X\,\seq t\right)$
      and some occurrence of the $i^{\text{th}}$ parameter of $\varrho X$
      is base-simple, then $\varrho s_{i}=\varrho t_{i}$.
      \end{lem}
      \begin{proof}
      By Lemma~\ref{lem:base-simple->inj.}, $\varrho s_{i}$ occurs in
      $\varrho X\left(\varrho\seq s\right)$ at certain position that depends
      only on $\varrho X$. Similarly $\varrho t_{i}$ occurs in $\varrho X\left(\varrho\seq t\right)=\varrho X\left(\varrho\seq s\right)$
      at the same position, and hence $\varrho s_{i}=\varrho t_{i}$.
      \end{proof}

    We define more properties to determine which binding to apply to a given constraint.
    Roughly speaking, the simple comparison form will trigger identification bindings,
    projectivity will trigger Huet-style projections, and
    simple projectivity will trigger JP-style projections.
    \begin{defi}[$\jp D4$]%
    \looseness=-1%
    \label{def:simple-comparison-form}We say that $s$ and $t$ are in \emph{simple
    comparison form }if all $\omega$-simple heads of opponent pairs in $s$ and $t$ are distinct, and
    each opponent pair has an $\omega$-simple head.
    \end{defi}
    \begin{defi}[$\jp D5$]%
    \label{def:projective} A term $t$ is called \emph{projective }if the head
    of $t$ is a parameter of $t$. If the whole body is just
    the parameter, then $t$ is called \emph{simply projective.}
    \end{defi}
    A central part of the proof is to find a suitable measure for the
    remaining problem size. Showing that the measure is strictly decreasing
    and well-founded guarantees that the procedure finds
    a suitable substitution in finitely many steps. We reuse the measure for remainder substitutions
    from JP~\cite{jp-unif-76}, but embed it into a lexicographic measure
    to handle the decomposition steps and oracles of our procedure.
    \begin{defi}[$\jp D7$]\label{def:weights}
    The \emph{free weight} of a term $t$ is the total number of occurrences of
    free variables and constants in $t$.
    The \emph{bound weight} of $t$ is the total number of occurrences
    (excluding occurrences $\lambda x$) of bound variables in $t$,
    but with the particular exemption: if a prefix variable $u$
    has one or more $\omega$-simple occurrences in the body,
    then one such occurrence and its arguments are not counted.
    It does not matter which occurrence is not counted because
    in $\eta$-long form the bound weight of the arguments of an
    $\omega$-simple variable is the same for all occurrences
    of that variable.
    \end{defi}
    \begin{defi}[$\jp D8$]\label{def:ord}
    For multisets $E$ of unification constraints and
    substitutions $\rho$, our measure on pairs $(E, \rho)$
    is the lexicographic comparison of
    \begin{enumerate}[{\bf A:}]
      \item the sum of the sizes of the terms in $\rho E$
      \item the sum over the free weight of $\rho F$,
        for all variables $F$ mapped by $\rho$
      \item the sum over the bound weight of $\rho F$,
      for all variables $F$ mapped by $\rho$
      \item the sum over the number of parameters of $\rho F$,
      for all variables $F$ mapped by $\rho$
    \end{enumerate}
    We denote the quadruple containing these numbers as $\ord(E,\rho)$.
    We denote the triple containing only the last three components of $\ord(E,\rho)$ as $\ord \rho$.
    We write $<$ for the lexicographic comparison of these tuples.
    \end{defi}
    %
    The next six lemmas correspond to the bindings of our procedure and sufficient
    conditions for the binding to bring us closer to a given solution.
    This is expressed as a decrease of the $\ord$ measure
    of the remainder. In each of these lemmas,
    let $u$ be a term with a variable head $a$ and
    $v$ a term with an arbitrary head $b$.
    Let $\varrho$ be a unifier of $u$ and $v$.
    The conclusion, let us call it \textbf{\emph{C}},
    is always the same: there exists a binding $\delta$ applicable to the problem $ u\unif v$, and
    there exists a substitution $\varrho'$ such that $\ord\varrho'<\ord\varrho$
    and for all variables $X$ except the fresh variables introduced by the binding $\varrho X= \varrho'\, \delta X$.
    For most of these lemmas, we refer to JP~\cite{jp-unif-76} for proofs.
    Although JP only claim $\varrho X= \varrho'\, \delta X$
    for variables $X$ mapped by $\varrho$,
    inspection of their proofs shows that the equality holds for
    all $X$ except the fresh variables introduced by the binding.
    Moreover, some of our bindings have more preconditions,
    yielding additional orthogonal hypotheses in our lemmas,
    which we address below.
    \begin{lem}
    [$\jp L9$]\label{lem:elimination}If $a=b$ is not an elimination variable and $\varrho a$
    discards any of its parameters, then \textbf{C} by elimination.
    Moreover, for the elimination variable $G$ introduced by this elimination,
    $\varrho'G$ discards none of its parameters and has the body of $\varrho a$.
    \end{lem}
    \begin{proof}
    Let $\varrho a=\param t$ and let ${\left(x_{j_{k}}\right)}_{k=1}^{i}$
    be the subsequence of $\seq x$ consisting of those variables which
    occur in the body $t$. It is a strict subsequence,\emph{ }since $\varrho a$
    is assumed to discard some parameter. Since the equal heads $a=b$ of the
    constraint $u\unif v$ are not elimination variables, elimination
    for ${\left(j_{k}\right)}_{k=1}^{i}$ can be applied. Let $\delta=\left\{ a\mapsto\param G\,x_{j_{1}}\,\ldots\,x_{j_{i}}\right\} $
    be the corresponding binding. Define $\varrho'$ to be like $\varrho$
    except
    \[
    \varrho'a=a\quad\text{and}\quad\varrho'G=\param[x_{j_{1}} \dots x_{j_{i}}]t.
    \]
    Obviously $\varrho'G$ is a closed term and $\varrho X = \varrho'\,\delta X$
    holds for all $X \not= G$. Moreover $\ord\varrho'<\ord\varrho$, because
    free and bound weights stay the same ($\varrho a$ and $\varrho'G$ have the
    same body $t$) whereas the number of parameters strictly decreases. The definition
    of ${\left(j_{k}\right)}_{k=1}^{i}$ implies that $\varrho'G$ discards
    none of its parameters.
    \end{proof}

    \begin{lem}
    [$\jp L{10}$]\label{lem:iteration}Assume that
    there exists a parameter $x$ of $\varrho a$ such that
    $x$ has a non-$\omega$-simple (Definition~\ref{def:omega-simple})
    occurrence in $\varrho a$, which is not below another parameter, or
    such that $x$ has at least two $\omega$-simple occurrences in $\varrho a$.
    Moreover,
    if $a = b$, to make iteration applicable,
    $a$~must not be an elimination variable, and
    $x$ must be of functional type.
    Then \textbf{C} is achieved by iteration.
    \end{lem}
    \begin{lem}
    [$\jp L{11}$]\label{lem:jp-projection}
    Assume that $a$ and $b$ are different free variables.
    If $\varrho a$ is simply
    projective (Definition~\ref{def:projective})
    and $a$ is not an identification variable,
    then \textbf{C} by JP-style projection.
    \end{lem}
    %
    \begin{lem}
    [$\jp L{12}$]\label{lem:jp-imitation}If $\varrho a$ is not projective and $b$ is rigid,
    then \textbf{C} by imitation.
    \end{lem}
    \begin{lem}
    [$\jp L{13}$]\label{lem:identification}
    Let $a \not= b$.
    Assume that $\varrho a \not= a$ and $\varrho b \not= b$
    are in simple comparison form (Definition~\ref{def:simple-comparison-form}) and neither is
    projective.
    Then \textbf{C} by identification.
    Moreover, $\varrho' H$ is not projective, where $H$ is the identification variable
    introduced by this application of the identification binding.
    \end{lem}
    \begin{proof}
    This is JP's Lemma~3.13, plus the claim that $\varrho' H$ is not projective.
    Inspecting the proof of that lemma, it is obvious that $\varrho' H$ cannot be projective
    because $\varrho a$ and $\varrho b$ are not projective.
    \end{proof}
    %
    %
    \begin{lem}%
    \label{lem:huet-projection}
    Assume that $\varrho a$ is
    projective (Definition~\ref{def:projective}),
    $a$ is not an identification variable, and $b$ is rigid.
    Then \textbf{C} by Huet-style projection.
    \end{lem}
    \begin{proof}
    Since $\varrho a$ is
    projective, we have $\varrho a = \lambda \overline{x}_n.\, x_k\,\overline{t}_m$
    for some $k$ and some terms $\overline{t}_m$.
    If $\varrho a$ is also simply projective,
    then $x_k$ must be non-functional
    and since Huet-style projection and JP-style projection coincide
    in that case, Lemma~\ref{lem:jp-projection} applies.
    Hence, in the following we may assume that $\varrho a$ is not simply projective,
    i.e., that $m>0$.

    Let $\delta$ be
    the Huet-style projection binding:
    \[ \delta = \{a \mapsto \lambda \overline{x}_n. \, x_i \, (F_1 \, \overline{x}_n) \, \ldots \, (F_m \, \overline{x}_n)\} \]
    for fresh variables $F_1,\dots,F_m$.
    This binding is applicable because $b$ is rigid.
    Let $\varrho'$ be the same as $\varrho$ except that we set $\varrho' a = a$
    and for each $1\leq j\leq m$ we set
    \[
      \varrho' F_j = \lambda \overline{x}_n. \, t_j
    \]
    It remains to show that $\ord\varrho'<\ord\varrho$.
    The free weight of $\varrho a$ is the same as
    the sum of the free weights of $\varrho' F_j$ for $1\leq j\leq m$.
    Thus, the free weight is the same for $\varrho$ and $\varrho'$.
    The bound weight of $\varrho a$ however is exactly $1$ larger than
    the sum of the bound weights of $\varrho' F_j$ for $1\leq j\leq m$
    because of the additional occurrence of $x_k$ in $\varrho a$.
    The exemption for $\omega$-simple occurrences in the definition
    of the bound weight cannot be triggered by this occurrence of $x_k$
    because $m > 0$ and thus $x_k$ is not $\omega$-simple.
    It follows that $\ord\varrho'<\ord\varrho$.
    \end{proof}
    We are now ready to prove the completeness theorem (Theorem~\ref{thm:completeness}).
    \begin{proof}
    Let $E$ be a multiset of constraints
    and let $V$ be the supply of fresh variables provided to our procedure.
    Let $\varrho$ be a unifier of $E$.
    We must show that there exists
    a derivation $\left(E,\makeop{id}\right)\longrightarrow^{*}\sigma$ and a
    substitution $\theta$ such that for all free variables $X\not\in V$, we have $\varrho X = \theta\sigma X$.

    Let $E_{0}=E$ and $\sigma_{0}=\makeop{id}$.
    Let $\varrho_{0}=\tau\varrho$ for some renaming $\tau$, such that
    every free variable occurring in $\varrho_{0}E_{0}$ does not occur in $E_{0}$ and is not contained in $V$.
    Then $\varrho_{0}$
    unifies $E_{0}$ because $\varrho$ unifies $E$ by assumption.
    Moreover, $\varrho_{0}=\varrho_{0}\sigma_{0}$. We proceed
    to inductively define $E_{j}$, $\sigma_{j}$ and $\varrho_{j}$
    until we reach some $j$ such that $E_j = \varnothing$.
    To guarantee well-foundedness, we ensure that the measure
    $\ord\left(\varrho_{j},E_{j}\right)$ decreases with each step.
    We maintain the following invariants for all $j$:
    \begin{itemize}
    \item $\left(E_{j},\sigma_{j}\right)\longrightarrow\left(E_{j+1},\sigma_{j+1}\right)$;
    \item $\ord\left(\varrho_{j},E_{j}\right) > \ord\left(\varrho_{j+1},E_{j+1}\right)$;
    \item $\varrho_{j}$ unifies $E_{j}$;
    \item $\varrho_{0} X = \varrho_{j}\sigma_{j} X$ for all free variables $X\not\in V$;
    \item every free variable occurring in $\varrho_{j}E_{j}$ does not occur in $E_{j}$ and is not contained in $V$;
    \item for every identification variable $X$, $\rho_j X$ is not projective; and
    \item for every elimination variable $X$, each parameter of $\rho_j X$
          has occurrences in $\rho_j X$,
          all of which are base-simple.
    \end{itemize}
    If $E_j \not= \varnothing$, let $u\unif v$ be the selected constraint $S(E_{j})$ in $E_{j}$.

    First assume that an oracle is able to find a CSU for the constraint $u\unif v$.
    Since $\varrho_{j}$ unifies $u$ and~$v$, by the definition of a CSU,
    the CSU discovered by the oracle
    contains a
    unifier $\delta$ of $u$ and $v$ such that there exists $\varrho_{j+1}$
    and for all free variables $X$ except for the auxiliary variables of the CSU
    we have $\varrho_{j} X=\varrho_{j+1}\,\delta X$. Thus, an \rulename{OracleSucc}
    transition is applicable and yields the node
    $\left(E_{j+1},\sigma_{j+1}\right)=\left(\delta\left(E_{j}\setminus\left\{ u\unif v\right\} \right),\delta\,\sigma_{j}\right)$.
    Therefore we have a strict containment $\varrho_{j+1}E_{j+1}\subset\varrho_{j+1}\,\delta\,E_{j}=\varrho_{j}E_{j}$.
    This implies $\ord\left(E_{j+1},\varrho_{j+1}\right)<\ord\left(E_{j},\varrho_{j}\right)$.
    It also shows that the constraints $\varrho_{j+1}E_{j+1}$ are unified
    when $\varrho_{j}E_{j}$ are.
    Since the auxiliary variables introduced by \rulename{OracleSucc} are fresh,
    they cannot occur in $E_j$ nor in $\sigma_{j} X$ for any $X\not\in V$. Hence,
    we have $\varrho_{0} X = \varrho_{j}\sigma_{j} X = \varrho_{j+1}\,\delta\,\sigma_{j} X =\varrho_{j+1}\sigma_{j+1} X$
    for all free variables $X\not\in V$.
    Any free variable occurring in $\varrho_{j+1}E_{j+1}$ cannot not occur in $E_{j+1}$ and is not contained in $V$
    because  $\varrho_{j+1}E_{j+1}\subset\varrho_{j}E_{j}$
    and the variables in $E_{j+1} = \delta\left(E_{j}\setminus\left\{ u\unif v\right\} \right)$
    are either variables already present in $E_{j}$ or fresh variables introduced by \rulename{OracleSucc}.
    New identification or elimination variables are not introduced;
    so their properties are preserved. Hence all invariants are preserved.

    Otherwise we proceed by a case distinction on the form of $u\unif v$.
    Typically, one of the Lemmas~\ref{lem:elimination}\textendash{}\ref{lem:huet-projection}
    is going to apply. Any one of them gives substitutions $\varrho'$
    and $\delta$ with properties that let us define $E_{j+1}=\delta E_{j}$,
    $\sigma_{j+1}=\delta\,\sigma_{j}$ and $\varrho_{j+1}=\varrho'$.
    The problem size always strictly decreases, because
    these lemmas imply
    $\varrho_{j+1}E_{j+1}=\varrho_{j+1}\delta E_{j}=\varrho_{j}E_{j}$
    and $\ord\varrho_{j+1}=\ord\varrho'<\ord\varrho_{j}$.
    Regarding the other invariants, the former equation guarantees that $\varrho_{j+1}$
    unifies $E_{j+1}$, and
    $\varrho_{0} X = \varrho_{j}\,\sigma_{j} X = \varrho_{j+1}\,\delta\,\sigma_{j} X =\varrho_{j+1}\,\sigma_{j+1} X$
    for all $X \not\in V$ because the fresh variables introduced by the binding cannot occur in $\sigma_{j} X$ for any $X\not\in V$. The conditions on variables must be checked
    separately when new ones are introduced. Let $a$ be
    the head of $u=\param a\,\seq u$ and $b$ be the head of $v = \param b\,\seq v$.
    Consider the following cases:

    \begin{description}[leftmargin=3mm]
    \item[$u$ and $v$ have the same head symbol $a=b$] \hfill
    \begin{enumerate}[leftmargin=6ex]
    \item Suppose that $\varrho_{j}a$ has a parameter with non-base-simple
    occurrence. By one of the induction invariants, $a$ is not an elimination
    variable.
    Among all non-base-simple occurrences of parameters in $\varrho_{j}a$,
    choose the leftmost one, which we call $x$.
    This occurrence of $x$ cannot be below another parameter,
    because having $x$ occur in one of its arguments would make
    that other parameter non-base-simple, contradicting the occurrence of $x$ being leftmost.
    Thus $x$ is neither base-simple nor below another parameter; so $x$ is of functional type.
    Moreover, non-base-simplicity implies non-$\omega$-simplicity.
    Hence, we can apply Lemma~\ref{lem:iteration} (iteration).

    \item Otherwise suppose that $\varrho_{j}a$ discards some of its parameters.
    By one of the induction invariants, $\varrho_{j}a$ is not an elimination variable.
    Hence Lemma~\ref{lem:elimination} (elimination) applies. The newly introduced elimination
    variable $G$ satisfies the required invariants, because Lemma~\ref{lem:elimination}
    guarantees that $\varrho_{j+1}G$ uses its parameters and shares the
    body with $\varrho_{j}a$ which by assumption of this case contains
    only base-simple occurrences.

    \item Otherwise every parameter of $\varrho_{j}a$ has occurrences and
    all of them are base-simple. We are going to show that $\rulename{Decompose}$
    is a valid transition and decreases $\varrho_{j}E_{j}$. By Lemma~\ref{lem:eq-if-omega-simple}
    we conclude from $\varrho_{j}u=\varrho_{j}v$ that $\varrho_{j}u_{i}=\varrho_{j}v_{i}$
    for every $i$. Hence the new constraints
    $E_{j+1}=E_{j}\setminus\left\{ u\unif v\right\} \cup\set{u_{i}\unif v_{i}}{\text{for all }i}$
    after $\rulename{Decompose}$ are unified by $\varrho_{j}$.
    This allows us to define $\varrho_{j+1}=\varrho_{j}$ and
    $\sigma_{j+1}=\sigma_{j}$.
    To check that $\varrho_{j+1}E_{j+1}=\varrho_{j}E_{j+1}$ is smaller
    than $\varrho_{j}E_{j}$ it suffices to check that constraints $\varrho_{j}u_{i}\unif\varrho_{j}v_{i}$
    together are smaller than $\varrho_{j}u\unif\varrho_{j}v$. Since
    all parameters of $\varrho_{j}a$ have base-simple occurrences, $\varrho_{j}u_{i}$
    is a subterm of $\varrho_{j}u=\param\varrho_{j}a\,\left(\varrho_{j}\seq u\right)$
    by Lemma~\ref{lem:base-simple->inj.}. Similarly for $\varrho_{j}v$.
    It follows that $\varrho_{j+1}E_{j+1}$ is smaller
    than $\varrho_{j}E_{j}$.
    Since $\varrho_{j+1}=\varrho_{j}$ and
    $\sigma_{j+1}=\sigma_{j}$, the other
    invariants are obviously preserved.
    \end{enumerate}

    \item[\mbox{$ u$} and \mbox{$ v$} is a flex-flex pair with different heads] \hfill
    \begin{enumerate}[leftmargin=6ex]
      \setcounter{enumi}{4}
      \item First, suppose that
      $\varrho_j a$ or $\varrho_j b$ is simply projective
      (Definition~\ref{def:projective}).
      By the induction hypothesis, the simply projective head
      cannot be an identification variable.
      Thus Lemma~\ref{lem:jp-projection} (JP-style projection) applies.

      \item Otherwise suppose that
      $\varrho_j a$ is projective but not simply.
      Then the head of $\varrho_j a$ is some parameter $x_{k}$. But
      this occurrence cannot be $\omega$-simple because it has arguments
      which cannot be bound above the head $x_{k}$. Thus Lemma~\ref{lem:iteration}
      (iteration) applies.
      If $\varrho_j b$ is projective but not simply, the same argument applies.

      \item Otherwise suppose that $\varrho_j a$, $\varrho_j b$ are in simple
      comparison form (Definition~\ref{def:simple-comparison-form}).
      By one of the the induction invariants,
      the free variables occurring in $\varrho_{j}E_{j}$ do not occur in $E_{j}$.
      Thus $\varrho_{j} a \not= a$ and $\varrho_{j} b \not= b$.
      Then Lemma~\ref{lem:identification} (identification) applies.

      \item Otherwise $\varrho_j a$, $\varrho_j b$ are not in simple comparison
      form. By Lemma~\ref{lem:principal-lemma} and by the definition of simple comparison form,
      there is some opponent pair $x_{k}\,\seq r$, $ b$
      in $\varrho_j a$ and $\varrho_j b$
      (after possibly swapping $u$ and $v$) where either the occurrence
      of $x_{k}$ is not $\omega$-simple (Definition~\ref{def:omega-simple}) or else
      $x_{k}$ has another $\omega$-simple occurrence in the body of $\varrho_j a$.
      Then Lemma~\ref{lem:iteration}
      (iteration) applies.
    \end{enumerate}

    \looseness=-1
    \item[$u$ and $v$ is a flex-rigid pair]
    Without loss of generality, assume that $a$ is flex and $b$ is rigid.\hfill

    \begin{enumerate}[leftmargin=6ex]
    \setcounter{enumi}{8}
    \item Suppose first that $\varrho_j a$ is projective.
    By one of the induction invariants, $a$ cannot be an identification variable.
    Thus Lemma~\ref{lem:huet-projection} (Huet-style projection) applies.

    \item Otherwise $\varrho_j a$ is not projective.
    The head of $\varrho_j a$ must be $b$ because $b$ is rigid, and $\varrho_j$ unifies $u$ and $v$.
    Since $\varrho_j a$ is not projective, that means that $b$ is not a bound variable.
    Therefore, $b$ must be a constant.
    Then Lemma~\ref{lem:jp-imitation} (imitation) applies.
    \end{enumerate}

    We have now constructed a run
    $\left(E_{0},\sigma_{0}\right)\longrightarrow\left(E_{1},\sigma_{1}\right)\longrightarrow\left(E_{2},\sigma_{2}\right)\longrightarrow\cdots$
    of the procedure.
    This run cannot be infinite because because the measure $\ord\left(E_j,\varrho_{j}\right)$
    strictly decreases as $j$ increases.
    Hence, at some point we reach a $j$ such
    that $E_{j}=\varnothing$ and $\varrho_{0} X = \varrho_{j}\,\sigma_{j} X$ for all $X \not\in V$.
    Therefore,
    $\left(E,\makeop{id}\right)\longrightarrow^*\left(\varnothing,\sigma_{j}\right)\longrightarrow\sigma_j$,
    and $\rho X = \tau^{-1}\,\varrho_{j}\,\sigma_{j} X$ for all $X \not\in V$,
    completing the proof. \qedhere
    \end{description}
    \end{proof}

\section{A New Decidable Fragment}%
\label{sec:solid-oracle}
\looseness=-1
We discovered a new fragment that admits a finite CSU and a simple oracle. The
oracle is based on work by Prehofer and the PT procedure~\cite{cp-unifphd-95},
an adaptation of the preunification procedure by Snyder and Gallier~\cite{sg-unif-89} (which itself is an adaptation of Huet's procedure).
PT transforms an initial multiset of constraints $E_0$
by applying bindings $\varrho$.
If there is a sequence $E_0\ptarrow{\varrho_1}\cdots\ptarrow{\varrho_n}E_n$
such that $E_n$ has only flex-flex constraints,
we say that PT produces a preunifier $\sigma = \varrho_n\ldots\varrho_1$
with constraints $E_n$.
A sequence fails if $E_n=\bot$.
As in the previous section, we consider
all terms to be $\alpha\beta\eta$-equivalence classes
with the $\eta$-long $\beta$-reduced form
as their canonical representative.
Unlike previously, in this section we view
unification constraints $s \unif t$ as ordered pairs.

The following rules, however, are stated modulo orientation.
The PT transition rules, adapted for our presentation style, are as follows:
\pagebreak[2]
\begin{description}
    \item[\rm\textsf{Deletion}]
      $\{\selected{s \unif s}\} \uplus E \ptarrowid E$
    \item[\rm\textsf{Decomposition}]
      $\{\selected{\lambda \overline{x}. \, a \, \overline{s}_m \unif \lambda \overline{x}. \, a \, \overline{t}_m} \}
        \uplus \allowbreak
        E \ptarrowid \{ s_1 \unif  t_1, \ldots, s_m \unif t_m \} \uplus E$
        \\ where $a$ is rigid
    \item[\rm\textsf{Failure}]
        $\{\selected{\lambda \overline{x}. \, a \, \overline{s} \unif \lambda \overline{x}. \, b \, \overline{t}} \}
          \uplus \allowbreak
          E \ptarrowid \bot$
    \\ where $a$ and $b$ are different rigid heads
    \item[\rm\textsf{Solution}]
      $\{\selected{\lambda \overline{x}.\, F\,\overline{x} \unif \lambda \overline{x}.\,t} \} \uplus E \ptarrow{\varrho} \varrho E$
      \\where $F$ does not occur in $t$, $t$ does not have a flex head, and $\varrho = \{ F \mapsto \lambda \overline{x}.\,t \}$
    \item[\rm\textsf{Imitation}]
      $\{\selected{\lambda \overline{x}. \, F \, \overline{s}_m \unif \lambda
      \overline{x}. \, \cst{f} \, \overline{t}_n }\} \uplus E \ptarrow{\varrho}
      \varrho(\{ G_1 \,
      \overline{s}_m  \unif  t_1, \ldots,  G_n \, \overline{s}_m \allowbreak \unif t_n\} \uplus E)$ \\
      where $\varrho = \{ F \mapsto \lambda \overline{x}_m. \, \cst{f} \, (G_1
      \, \overline{x}_m) \ldots (G_n \, \overline{x}_m) \}$,
      $\overline{G}_n$ are fresh variables of appropriate types
    \item[\rm\textsf{Projection}]
      $\{\selected{\lambda \overline{x}. \, F \, \overline{s}_m \unif \lambda
      \overline{x}. \, a \, \overline{t} }\} \uplus E \ptarrow{\varrho}
      \varrho(\{ s_i \, (G_1 \,
      \overline{s}_m) \ldots (G_j \, \overline{s}_m) \unif a \, \overline{t}\} \uplus E)$ \\
      where $\varrho =\{F \mapsto \lambda \overline{x}_m. \, x_i \, (G_1 \,
      \overline{x}_m) \ldots (G_j \, \overline{x}_m)\}$, $\overline{G}_j$ are
      fresh variables of appropriate types
\end{description}
\looseness=-1
The \selected{\text{grayed}} constraints are required to be selected by a given
selection function $S$. We call $S$ \emph{admissible} if it selects only flex-rigid constraints, prioritizes
selection of constraints applicable for \textsf{Failure} and \textsf{Decomposition},
and of descendant constraints of
\textsf{Projection} transitions with $j=0$ (i.e., for $x_i$ of base type), in that order of priority. In the
remainder of this section we consider only admissible selection functions,
an assumption that Prehofer also makes implicitly in his thesis. Additionally,
whenever we compare multisets, we use the multiset ordering defined by
Dershowitz and Manna~\cite{dm-multisetordering-79}.
As above, we assume that the fresh variables are taken from an infinite supply $V$ of fresh variables
that are different from the variables in the initial problem
and never reused.

The following lemma states that PT is \emph{complete for preunification}:
\begin{lem}%
  \label{lemma:pt-complete}
  \looseness=-1
  Let $\varrho$ be a unifier of a multiset of constraints $E_0$. Then PT produces a
  preunifier $\sigma$ with constraints $E_n$, and there exists a
  unifier $\theta$ of $E_n$ such that
  $\varrho X = \theta \sigma X$ for all $X$ that are not contained in the supply $V$ of fresh variables.
\end{lem}
\begin{proof}
  This lemma is a refinement of Lemma 4.1.7 from Prehofer's PhD thesis~\cite{cp-unifphd-95}, and this proof closely follows the proof of that lemma.
  Compared to the lemma from Prehofer's thesis, our lemma additionally establishes
  the relationship of unifier $\theta$ of the resulting flex-flex constraint
  set $E_n$ with preunifier $\sigma$ and the target unifier $\varrho$.

  We prove the lemma by induction
  using a well-founded measure on $(\varrho,
  E_0)$. The ordering
  is the lexicographic comparison of the following properties:

  \begin{enumerate}[{\bf A:}]
    \item sum of the abstraction-free sizes of the terms $\varrho F$ for each variable $F$ mapped by $\varrho$
    \item multiset of the sizes of constraints in $E_0$
  \end{enumerate}
Here, the \emph{abstraction-free size} is inductively defined by
$\text{afsize}(F) = 1$; $\text{afsize}(x) = 1$; $\text{afsize}(\cst{f}) = 1$;
$\text{afsize}(s\,t) = \text{afsize}(s) + \text{afsize}(t)$;
$\text{afsize}(\lambda x.\, s) = \text{afsize}(s)$.

  If $E_0$ consists only of flex-flex constraints, then we can take an empty
  transition sequence (i.e., $\sigma = \text{id}$) and $\theta = \varrho$.
  Otherwise, there exists a constraint that is selected by an
  admissible selection function. We show that for each such constraint, there is
  going to be a PT transition bringing us closer to the desired preunifier.

  Let $E_0 = \{\selected{s \unif t}\} \uplus E'_0$. Since $s \unif t$ is selected, at least
  one of $s$ and $t$ must have a rigid head. We distinguish several cases based on the
  form of $s \unif t$ (modulo the constraint order):

  \begin{itemize}
    \item $s \unif s$: in this case, \textsf{Deletion} applies. This transition does not
    alter $\varrho$, but removes an equation obtaining $E_1$ which is smaller
    than $E_0$ and still unifiable by $\varrho$. By induction hypothesis, the preunifier is reachable.
    \item
    \looseness=-1
    $\lambda \overline{x}. \, a \, \overline{s}_m \unif \lambda
    \overline{x}. \, b \, \overline{t}_n$, where $a$ and $b$ are rigid: since
    $\varrho$ is a unifier, then $a=b$, $n=m$, and \textsf{Decomposition}
    applies. Similarly to the above case, we conclude that the preunifier is reachable.
    \item $\lambda \overline{x}.\, F\,\overline{x} \unif \lambda \overline{x}.\,t$ where $t$
    has a rigid head and $F$ does not occur in $t$: \textsf{Solution} applies. Huet showed~\cite[proof of L5.1]{gh-unification-75}
    that in this case $\varrho = \varrho\sigma_1$, where $\sigma_1 = \{ F \mapsto \lambda\overline{x}.\,t \}$.
    Let $\varrho_1 = \varrho\setminus\{F \mapsto \varrho F\}$.
    Then $\varrho = \varrho_1 \sigma_1$.
    \looseness=-1
    Since $\varrho$ unifies $E'_0$, $\varrho_1$ unifies $E_1 = \sigma_1 E'_0$. Clearly, $\varrho_1$ is smaller
    than $\varrho$. Now we can apply induction hypothesis on $\varrho_1$ and $E_1$, from which
    we obtain a sequence $E_1 \ptarrow{\sigma_2} \cdots \ptarrow{\sigma_n} E_n$, $\theta$ which
    unifies $E_n$ and $\sigma' = \sigma_n \ldots \sigma_2$, such that $\varrho_1 X = \theta\sigma' X$ for all $X \not\in V$.
    Finally, it is clear that sequence $E_0 \ptarrow{\sigma_1} E_1 \ptarrow{\sigma_2} \cdots \ptarrow{\sigma_n} E_n$,
    $\theta$, and $\sigma = \sigma' \sigma_1$ are as wanted in the lemma statement.
    \item \looseness=-1
    $\lambda
    \overline{x}.\,F\,\overline{s}_m~\unif~\lambda \overline{x}.\,a \,
    \overline{t}$ where $a$ is rigid: depending on the value of $\varrho$, we can
    either take an \textsf{Imitation} or \textsf{Projection} step. In either case, we show that
    the measure reduces. Let $\varrho F = \lambda \overline{x}_m. \, b \, \overline{u}_n$ where $b$
    is either a constant or a bound variable. If $b$ is a constant, we choose
    the \textsf{Imitation} step; otherwise we choose the \textsf{Projection} step. In either case,
    $\sigma_1 = \{ F \mapsto \lambda \overline{x}_m. \, b \, (G_1 \, \overline{x}_m) \ldots
    (G_n \, \overline{x}_m) \}$. Then it is easy to check that $\varrho_1 = \varrho
    \setminus \{F \mapsto \lambda \overline{x}_m. \, b \, \overline{u}_n\} \cup \{ G_1~\mapsto~\lambda\overline{x}_m.\, u_1, \allowbreak
    \ldots, G_n \mapsto \lambda\overline{x}_m.\, u_n \}$ unifies $E_1$ created as the
    result of imitation or projection. Clearly, $\varrho_1$ has smaller abstraction-free size
    than $\varrho$. Therefore, by the induction hypothesis we obtain the transitions
    $E_1 \ptarrow{\sigma_2} \cdots \ptarrow{\sigma_n} E_n$ and substitutions
    $\theta$ and  $\sigma' = \sigma_n\ldots\sigma_2$, where $\varrho_1 X =
    \theta \sigma' X$ for all $X \not\in V\setminus\{G_1,\dots,G_n\}$ and $\theta$ is a unifier of $E_n$. Our goal is to prove
    $\varrho X = \theta \sigma' \sigma_1 X$ for all $X \not\in V$.
    For variables $X \not\in V$ that are not $F$, we have $\sigma_1 X = X$
    and $\varrho X = \varrho_1 X$, which implies $\varrho X = \theta \sigma' \sigma_1 X$.
    For $X = F$, since $\varrho_1$ and $\theta \sigma'$ agree on $G_1, \ldots, G_n$, we
    conclude that $\varrho F = \theta \sigma' \sigma_1 F$.
    Therefore, by taking $\sigma = \sigma'\sigma_1$, and prepending $E_0$ to the
    sequence from the induction hypothesis, we show that the preunifier is reachable.
    \qedhere
  \end{itemize}
\end{proof}

\looseness=-1
Prehofer showed that PT terminates for some classes of constraints. Those
classes impose requirements on the free variables occurring in the constraints.
In particular, he identified a class of terms we call \emph{strictly
solid}\footnote{In Prehofer's thesis this class is described in the statement of
Theorem 5.2.6.}, which requires that all arguments of free variables are either
bound variables (of arbitrary type) or ground second-order terms of base type.
Another important constraint is linearity: a term is called
\emph{linear} if it contains no repeated occurrences of free variables.

\begin{exa}%
\label{ex:solid-terms}
Let $G$, $\cst{a}$, and $x$ be of base type, and $F$, $H$, $\cst{g}$, and $y$ be
binary. Then, the term $F \, G \, \cst{a}$ is not strictly solid, since $F$ is
applied to a free variable $G$; similarly $H \, (\lambda x. \, x) \, \cst{a}$ is
not strictly solid as the first argument is of functional type, but it is not
a bound variable; $\lambda x. \, F \, x \, \cst{a}$ is strictly solid, but $F \, \cst{a} \, (\cst{g} \, (\lambda y. \,
y\,\cst{a}\,\cst{a})  \, \cst{a})$ is not, as the second
argument of $F$ is a ground term of third order.
\end{exa}

Prehofer's thesis states that PT terminates on $\{s \unif t\}$
if $s$ is linear, $s$ shares no free variables with~$t$, $s$ is strictly solid,
and $t$ is second-order. Together with completeness of PT for preunification,
this result implies that PT procedure can be used to enumerate finitely many
elements of complete set of preunifiers for this class of terms.

Prehofer focused on the preunification problem as he remarks that the resulting
flex-flex pairs are ``intricate''~\cite[Sect.~5.2.2]{cp-unifphd-95}. We discovered that these flex-flex
pairs actually have an MGU, allowing us to solve the full unification problem,
rather than preunification problem. Moreover, we lift the order restriction
imposed by Prehofer: We identify a class of terms called \emph{solid} which
requires that all arguments of free variables are either bound variables (of
arbitrary type) or ground (arbitrary-order) terms of base type.

\begin{exa}
  Consider the setting of Example~\ref{ex:solid-terms}. The terms $F \, G \,
  \cst{a}$ and $H \, (\lambda x. \, x) \, \cst{a}$ are not solid, for the same reasons they are not strictly solid.
  However, since the order restriction is lifted, $F \, \cst{a} \, (\cst{g} \, (\lambda y. \,
  y\,\cst{a}\,\cst{a})  \, \cst{a})$ is solid.
\end{exa}

\looseness=-1
In other words, we extend the above preunification decidability result
for linear strictly solid terms along two axes: we create an oracle
for the full unification problem and we lift the
order constraints.
To enumerate a CSU for a problem $E_0 = \{s \unif t\}$, where $s$ and $t$ are solid,
$s$ is linear and shares no free variables with~$t$, our oracle applies the following two steps:
\begin{enumerate}
  \item Apply the PT procedure on $E_0$ to obtain a preunifier $\sigma$ with
  flex-flex constraints $E_n$.
  \item\label{step:ff}\looseness=-1 If $E_n$ is empty, return $\sigma$. Otherwise, choose a flex-flex constraint $u \unif v$ from $E_n$, and let the MGU of $u \unif v$
  be $\varrho$. Then, set $E_n := \varrho(E_n \setminus \{ u \unif v \})$ and
   $\sigma := \varrho\sigma$, and repeat step~\ref{step:ff}.
\end{enumerate}

To show that this oracle terminates and yields a CSU, we must prove that the PT
procedure terminates on above described class of problems (Lemma~\ref{lemma:solid-termination-preunif}). Moreover, we must show how to compute MGUs for
the remaining flex-flex pairs (Lemma~\ref{lemma:solid-var-same-mgu} and~\ref{lemma:solid-vars-diff-mgu}). Our results are combined together in Theorem~\ref{thm:oracle}.

Towards proving that the PT
procedure terminates on above described class of problems, we first consider the corresponding matching problem.
A \emph{matching problem} is a unification problem
$\{ s \unif t \}$ where $t$ is ground. In what
follows, we establish some useful properties of matching problems in which both
$s$ and $t$ are solid (\emph{solid matching problems}). We call the unifier
of a matching problem a \emph{matcher}.

\begin{lem}%
    \label{lemma:pt-matching-termination}
    PT terminates on a solid matching problem $\{s \unif t\}$.
\end{lem}
\begin{proof}
    We show termination by designing a measure on a matching problem $E$ that
    decreases with each application of a PT transition (possibly followed by
    \textsf{Decomposition} or \textsf{Failure} steps).
    Our measure function
    compares the following properties lexicographically:
    \begin{enumerate}[{\bf A:}]
        \item multiset of the sizes of right-hand sides of the constraints in $E$
        \item number of free variables in $E$
    \end{enumerate}

    Clearly, when applying PT transitions, the problem stays a matching problem
    because the applied bindings do not introduce free variables on the
    right-hand sides. It is also easy to check that each applied binding keeps
    the terms in the solid fragment. Namely, bindings for both
    \textsf{Imitation} and \textsf{Projection} transitions are patterns, which means
    that, after applying the binding, fresh free variables are applied only to
    bound variables or ground base-type terms. The \textsf{Solution} transition effectively
    replaces the variable with a ground term, which is obviously solid. We show each transition
    either trivially terminates or reduces the measure:

    \begin{description}[leftmargin=3mm]
      \item[Deletion] A decreases.
      \item[Decomposition] A decreases.
      \item[Failure] Trivial---represents a terminal node.
      \item[Solution] This transition applies on
        constraints of the form
        $F \unif t$. This reduces A, since the constraint $F \unif t$ is removed
        and $F$ cannot appear on the right-hand side.
      \item[Imitation] The rule is applicable only on
        $\lambda \overline{x}.\, F \, \overline{s}_m  \unif \lambda
        \overline{x}.\, \cst{f} \, \overline{t}_n$. The \textsf{Imitation} transition
        replaces the constraint with $\{H_i \,
        \overline{s}_m  \unif \overline{t}_i \mid 1\le i\le n\}$. This reduces A, as $F$ does not appear on any right-hand side.
    \item[Projection] The rule is applicable only on
        $\lambda \overline{x}.\, F \, \overline{s}_m  \unif \lambda
        \overline{x}.\, a \, \overline{t}_n$. If $a$ is a bound variable,
        and we project $F$ to argument $s_i$ different than $a$, \textsf{Failure} applies and
        PT trivially terminates. If we project $F$ to $a$, we apply \textsf{Decomposition},
        which reduces $A$. If $a$ is a constant, for \textsf{Failure} not to apply, we
        have to project to a base-type ground term $s_i$. This does not increase A,
        since no variables appear on right-hand sides, but removes the variable $F$
        from $E$, reducing B by one. \qedhere
    \end{description}
\end{proof}

We say $\sigma$ is a grounding substitution if for every variable $F$ mapped
by $\sigma$, $\sigma F$ is ground.
\begin{lem}%
    \label{lemma:pt-matching-ground}
  All unifiers produced by PT for the solid matching problem $\{s \unif t\}$ are grounding substitutions.
\end{lem}
\begin{proof}
  Closely following the proof of Lemma 5.2.5 in Prehofer's PhD thesis~\cite{cp-unifphd-95}, we prove our claim by induction on the length of the
  PT transition sequence that leads to the unifier. We know this
  sequence is finite by Lemma~\ref{lemma:pt-matching-termination}. The base case of
  induction, for the empty sequence, is trivial. The induction step is made using
  one of the following transitions:

  \begin{description}[leftmargin=3mm]
    \item[Deletion] Trivial.
    \item[Decomposition] Trivial.
    \item[Failure] This rule is not relevant, since it will not lead to the unifier.
    \item[Solution] This rule applies the substitution $\{F \mapsto t\}$, which
    is grounding since $t$ is ground.
    \item[Imitation] This transition applies on
    constraints of the form
    \[ \lambda \overline{x}. \, F \, \overline{s}_k \unif \lambda
    \overline{x}. \, \cst{f} \, \overline{t}_l \]
    The binding for \textsf{Imitation} is
    $\varrho = \{F \mapsto \lambda \overline{x}_k.\, \cst{f}\, (G_1 \,
    \overline{s}_k ) \ldots (G_l \, \overline{s}_k ) \}$, and reduces the
    problem to $\{G_i \,  \overline{s}_k \unif t_i \mid 1\le i\le l\}$. Since the right-hand sides are ground,
    any unifier $\sigma$ produced by PT must map all of the variables $G_1, \ldots, G_l$. By induction hypothesis,
    $\sigma G_i$ is ground. Therefore, $\sigma\varrho$ must map $F$ to a ground term.
    \item[Projection] This transition applies on
    constraints of the form
    \[\lambda \overline{x}. \, F \, \overline{s}_k \unif \lambda
    \overline{x}. \, t\]
    The binding for \textsf{Projection} is
    $\varrho = \{F \mapsto \lambda \overline{x}_k. \, x_i \, (G_1 \,
      \overline{x}_k) \ldots (G_j \, \overline{x}_k)\}$,
    and reduces the problem to
    $\{s_i \, (G_1 \,
      \overline{s}_k) \ldots (G_j \, \overline{s}_k) \unif a \, \overline{t}_l\}$.
      Since the right-hand side is ground,
      any unifier $\sigma$ produced by PT must map all of the variables $G_1, \ldots, G_j$.
    By induction hypothesis,
    $\sigma G_i$ is ground. Therefore, $\sigma\varrho$ must map $F$ to a ground term. \qedhere
  \end{description}
\end{proof}

\begin{lem}%
  \label{lemma:solid-termination-preunif}
  \looseness=-1
  If $s$ and $t$ are solid, $s$ is linear and shares no free variables with $t$, then
  PT terminates for the preunification problem $\{ s \unif t \}$, and
  all remaining flex-flex constraints are solid.
\end{lem}
\begin{proof}
  This lemma is a modification of Lemma 5.2.1 and Theorem 5.2.6 from Prehofer's
  PhD thesis~\cite{cp-unifphd-95}. Correspondingly, the following proof closely
  follows the proofs for these two lemmas. We also adopt the notion of an
  \emph{isolated} variable from Prehofer's thesis: a variable is isolated in a
  multiset $E$ of constraints if it appears exactly once in $E$.

  \looseness=-1
  First, similarly to the proof of previous lemma we conclude each transition
  maintains the condition that the terms remain solid. Second, we have to show
  that variables on left-hand sides remain isolated. For \textsf{Imitation} and
  \textsf{Projection} rules, the preservation of this invariant is obvious.
  Later, we also show that \textsf{Solution} preserves it. Third, since $s$ and
  $t$ share no variables, and $s$ is linear, no rule can introduce a variable
  from the right-hand side to the left-hand side.

  To prove termination of PT, we devise a measure that
  decreases with each application of a PT transition. The measure
  lexicographically compares the following properties:
  \begin{enumerate}[{\bf A:}]
    \item number of occurrences of constant symbols and bound variables on left-hand sides that
    are not below free variables
    \item number of free variables on right-hand sides
    \item multiset of the sizes of right-hand sides
  \end{enumerate}

  We show that each transition either trivially terminates or reduces the measure:
  \begin{description}[leftmargin=3mm]
    \item[Deletion] A does not increase and at least one of A or B reduces.
    \item[Decomposition] A reduces.
    \item[Failure] Trivial.
    \item[Solution] This rule applies in two cases:

    \begin{itemize}
      \item $F \unif t$: since $F$ is isolated, A is unchanged, B is not
      increased, and C reduces. All variables on left-hand sides remain isolated
      since they are unaffected by this substitution.
      \item $t \unif F $: since $F$ does not occur on any left-hand side and the
      head of $t$ must be a constant or bound variable (otherwise the rule would
      not be applicable), $A$ reduces. Even though $t$ might contain some free
      variables and $F$ can have multiple occurrences on right-hand sides, all
      free variables on left-hand sides remain isolated. Namely, all free variables
      in $t$ occur exactly once in the multiset, and since $t \unif F$ is
      removed, all of them either disappear or end up on right-hand sides.
    \end{itemize}
    \item[Imitation] We distinguish the following two forms of the selected constraint:
    \begin{itemize}
      \item $\lambda \overline{x}. \, F \, \overline{v}_n \unif \lambda \overline{x}. \,
      \cst{f} \, \overline{u}_m$:
      We replace this constraint by
      $\{H_i \, \overline{v}_n \unif u_i \mid 1\le i\le m\}$
      and apply the \textsf{Imitation} binding
      $F \mapsto \lambda \overline{x}_n.\, \cst{f} \, (H_1 \, \overline{x}_n) \allowbreak
      \ldots (H_m \, \overline{x}_n)$.
      As $F$ is isolated, A and B do not increase.
      Since $F$ is isolated, it does not occur on any right-hand side, and hence C decreases.
      \item $\lambda \overline{x}. \, \cst{f} \, \overline{u}_m \unif \lambda \overline{x}. \, F \, \overline{v}_n $:
      applying the \textsf{Imitation} transition
      as above will reduce A since $F$ cannot appear on any
      left-hand side.
    \end{itemize}
    \item[Projection] Similarly to the previous transition rule, we have two cases:
      \begin{itemize}
        \item$\lambda \overline{x}. \, F \, \overline{v}_n \unif \lambda
        \overline{x}. \, a \, \overline{u}_m$: if $a$ is a bound variable,
        then for \textsf{Failure} not to be applicable afterwards, we have to
        project $F$ onto argument $v_i$ equal to $a$. Then we proceed like for
        \textsf{Imitation}. If $a$ is not a bound variable, then we have to
        project to some base-type term $v_j$ (otherwise \textsf{Failure} would
        be applicable afterwards). This reduces the problem to $\lambda
        \overline{x}. \, v_j \unif \lambda \overline{x}. \, a \,
        \overline{u}_m$. This is a solid matching problem, whose solutions
        computed by PT are grounding substitutions (see Lemma~\ref{lemma:pt-matching-ground}). Applying one of those solutions will
        eliminate all the free variables in $\lambda \overline{x}.\, a \,
        \overline{u}_m$. Since PT is parametrized by an admissible
        selection function, we know that there are no constraints descending
        from a simple projection in $E$ since the constraint $\lambda \overline{x}. \, F
        \, \overline{v}_n \unif \lambda \overline{x}. \, a \, \overline{u}_m$ was
        chosen,
        which, due to solidity restrictions, cannot be such a descendant.
        Therefore, we know that PT will transform the descendants
        of the matching problem $\lambda \overline{x}. \, v_j \unif \lambda
        \overline{x}.\, a \, \overline{u}_m$ until either \textsf{Failure} is
        observed (making PT trivially terminating) or until no descendant exists
        and the grounding matcher is computed (see Lemmas~\ref{lemma:pt-matching-ground} and~\ref{lemma:pt-matching-termination}).
        This results in removal of the original constraint $\lambda
        \overline{x}. \, F \, \overline{v}_n \unif \lambda \overline{x}. \,
        a \, \overline{u}_m$ and application of the computed grounding
        matcher, which will either remove all the free variables in the right-hand
        side of the constraint (not increasing A and reducing B)
        or not increasing A and B, reduce C if no free variables occur in the right-hand side.
        \item $\lambda \overline{x}. \, a \, \overline{v}_n \unif \lambda
        \overline{x}. \, F \, \overline{u}_m$:
        if $a$ is a bound variable, projecting $F$ onto argument $u_i$ will either enable
        application of \textsf{Decomposition} as the next step reducing A, or it will result in \textsf{Failure},
        trivially terminating. If $a$ is a constant, then projecting $F$ onto some $u_j$
        will either yield  \textsf{Failure} or enable \textsf{Decomposition}, reducing A. \qedhere
      \end{itemize}
    \end{description}
\end{proof}

Enumerating a CSU for a solid flex-flex pair may seem as hard as
for any other flex-flex pair; however,
the following two lemmas show that solid pairs admit an MGU\@:

\begin{lem}%
    \label{lemma:solid-var-same-mgu}
    The unification problem
    $\{\lambda \overline{z}.\, F\,\overline{s}_m \unif   \lambda \overline{z}.\, F\,\overline{s}'_{m}\}$, where both terms are
    solid, has an MGU of the form $\sigma=\{F \mapsto \lambda \overline{x}_m.\,
    G \, x_{j_1} \ldots x_{j_r}\}$ where $G$ is an auxiliary variable, and $1 \leq
    j_1 < \cdots < j_r \leq m$ are exactly those indices $j_i$ for which
    $s_{j_i} = s'_{j_i}$.
\end{lem}
\begin{proof}
    Let $\varrho$ be a unifier for the given unification problem.
    Let $\lambda \overline{x}.\, u = \varrho F$. Take an arbitrary subterm of $u$ whose head is a bound
    variable $x_i$. If $x_i$ is of function type, it corresponds to either $s_i$
    or $s'_i$ which, due to solidity restrictions, has to be a bound variable.
    Furthermore, since $\varrho$ is a unifier, $s_i$ and $s'_i$ have to be the
    syntactically equal. Similarly, if $x_i$ is of base type, it corresponds
    to two ground terms $s_i$ and $s'_i$ which have to be syntactically equal.
    We conclude that $\varrho$ can use variables from $\overline{x}_n$ only
    if they correspond to syntactically equal terms. Therefore,
    there is a substitution $\theta$ such that $\rho X = \theta\sigma X$ for all $X \not = G$.
    Due to arbitrary choice of $\varrho$, we conclude that $\sigma$ is an MGU\@.
\end{proof}

\begin{lem}%
  \label{lemma:solid-vars-diff-mgu}
  Let
  $\{\lambda \overline{x}.\, F \, \overline{s}_m \unif \lambda\overline{x}.\, F' \, \overline{s}'_{m'}\}$
  be a solid unification problem
  where $F \neq F'$.
  By Lemma~\ref{lemma:pt-matching-termination}, there exists a finite CSU $\{\sigma_i^1, \dots, \sigma_i^{k_i} \}$
  of the problem $\{s_i \unif H_i \, \overline{s}'_{m'}\}$,
  where $H_i$ is a fresh free variable.
  Let
  $\lambda\overline{y}_{m'}. \, s_i^j = \lambda\overline{y}_{m'}. \, \sigma_i^j(H_i) \,\overline{y}_{m'}$.
  Similarly, also by Lemma~\ref{lemma:pt-matching-termination},
  there exists a finite CSU $\{ \tilde\sigma_i^1, \dots, \tilde\sigma_i^{l_i} \}$
  of the problem $\{s'_i \unif \tilde H_i \, \overline{s}_m\}$, where $\tilde H_i$
  is a fresh free variable.
  Let
  $\lambda \overline{x}_m. \, {s'}_i^j = \lambda \overline{x}_m. \, \tilde\sigma_i^j(\tilde H_i) \, \overline{x}_m$.
  Let $Z$ be a fresh free variable.
  An MGU $\sigma$ for the given problem is
    \begin{align*}
      F &\mapsto \lambda \overline{x}_m.\, Z
          \, \underbrace{x_1 \, \ldots \, x_1}_{k_1 \text{ times}} \, \ldots \,
             \underbrace{x_m \, \ldots \, x_m}_{k_m \text{ times}}
             \, {s'}_1^1 \ldots {s'}_1^{l_1} \, \ldots \, {s'}_{m'}^1 \ldots {s'}_{m'}^{l_{m'}} \, \\
      F' &\mapsto \lambda \overline{y}_{m'}. \, Z
          \, s_1^1 \ldots s_1^{k_1} \ldots s_m^1 \ldots s_m^{k_m} \,
      \underbrace{y_1 \, \ldots y_1}_{l_1 \text{ times}} \, \ldots \,
      \underbrace{y_{m'} \, \ldots y_{m'}}_{l_{m'} \text{ times}}
    \end{align*}
  where the auxiliary variables are $H_1, \dots, H_m$, $\tilde H_1, \dots \tilde H_{m'}$, $Z$ and all auxiliary variables associated with the above CSUs.
\end{lem}
\begin{proof}
  Let $\varrho$ be an arbitrary unifier for the problem $\lambda
  \overline{x}.\, F \, \overline{s}_m \unif \lambda\overline{x}.\, F' \,
  \overline{s}'_{m'}$. We prove $\sigma$ is an MGU by showing that there exists a substitution
  $\theta$ such that $\varrho X = \theta\sigma X$ for all non-auxiliary variables $X$.
  We can focus only on $X \in \{F, F'\}$
  because all other non-auxiliary variables
  do neither appear in the original problem nor in $\sigma F$ or $\sigma F'$; so we can simply define $\theta X = \varrho X$.


  Let $\lambda \overline{x}_m. \, u = \varrho F$ and  $\lambda \overline{y}_{m'}. \, u' = \varrho{F'}$, where
  the bound variables $\overline{x}_m$ and $\overline{y}_{m'}$ have been $\alpha$-renamed apart. We also assume
  that the names of variables bound inside $u$ and $u'$ are $\alpha$-renamed so that they are different from $\overline{x}_m$ and
  $\overline{y}_{m'}$. Finally, bound variables
  from the definition of $\sigma$ have been $\alpha$-renamed to match $\overline{x}_m$ and $\overline{y}_{m'}$.
  We
  define $\theta$ to be the substitution
  \[\theta = \{ Z \mapsto \lambda z_1^1 \ldots z_1^{k_1} \ldots z_m^1 \ldots z_m^{k_m} \, w_1^{1} \ldots w_1^{l_1} \ldots w_{m'}^1 \ldots  w_{m'}^{l_{m'}}. \, \textsf{diff}(u, u') \}\]
  where \textsf{diff}$(v,v')$ is defined recursively as
  \begin{alignat}{3}
    &\textsf{diff}(\lambda x.\, v, \lambda y.\, v') &&{}={}&& \lambda x. \, \textsf{diff}(v, \{y \mapsto x\}v') \label{diff:alpha} \\
    &\textsf{diff}(a \, \overline{v}_n, a \, \overline{v}'_n) &&=&&  a \, \textsf{diff}(v_1, v'_1) \ldots \textsf{diff}(v_n, v'_n) \label{diff:samehd} \\
    &\textsf{diff}(x_i, v') &&=&&  z_i^k \text{, if $v' = {s}_i^k$ } \label{diff:xi}  \\
    &\textsf{diff}(v, y_i) &&=&&  w_i^l \text{, if $v = {s'}_i^l$ } \label{diff:yi}  \\
    &\textsf{diff}(x_i \, \overline{v}_n, y_j \, \overline{v}'_n) &&=&&  z_i^k \, \textsf{diff}(v_1, v'_1) \ldots \textsf{diff}(v_n, v'_n) \text{, if $y_j = s_i^k $ } \label{diff:xiyj}
  \end{alignat}

  \looseness=-1
  From \textsf{diff}'s definition it is clear that there are terms $v,v'$ for which it is undefined.
  However, we will show that for each $u$ and $u'$ that are bodies of bindings from a unifier $\varrho$,
  \textsf{diff} is defined and has the desired property.
  In equations~\ref{diff:xi},~\ref{diff:yi}, and~\ref{diff:xiyj}, if there are multiple numbers $k$ or $l$
  that fulfill the condition, choose an arbitrary one.
  We need to show that $\varrho F = \theta\sigma F$ and $\varrho F' = \theta\sigma F'$.
  By the definitions of $u$, $u'$, $\theta$ and $\sigma$ and $\beta$-reduction, this is equivalent to
  \begin{align*}
    \lambda \overline{x}_m.\, u &= \lambda \overline{x}_m.\,
    \{z_i^k \mapsto x_i, w_i^l \mapsto {s'}_i^l \text{ for all }k,l,i\}\,\textsf{diff}(u,u')\\
    \lambda \overline{y}_{m'}.\, u' &= \lambda \overline{y}_{m'}.\,
    \{z_i^k \mapsto s_i^k, w_i^l \mapsto y_i \text{ for all }k,l,i\}\,\textsf{diff}(u,u')
  \end{align*}
  We will show by induction that for any $\lambda \overline{x}_m.\, v$, $\lambda \overline{y}_m.\, v'$ such that
  \[\{x_1 \mapsto s_1, \ldots, x_m \mapsto  s_m\} v = \{y_1 \mapsto s'_1, \ldots, y_{m'} \mapsto  s'_{m'}\} v' \tag{$\star$}\]
  we have
  \begin{align*}
    \tag{\textdagger}
    \renewcommand{\arraystretch}{1.2}
    \begin{array}{c@{\qquad}l}
      v &=
      \{z_i^k \mapsto x_i, w_i^l \mapsto {s'}_i^l \text{ for all }k,l,i\}\,\textsf{diff}(v,v')\\
      v' &=
      \{z_i^k \mapsto s_i^k, w_i^l \mapsto y_i \text{ for all }k,l,i\}\,\textsf{diff}(v,v')
    \end{array}
  \end{align*}
  The equation ($\star$) holds for $v = u$ and $v' = u'$ because
  $\varrho$ is a unifier of
  $\lambda
  \overline{x}.\, F \, \overline{s}_m \unif \lambda\overline{x}.\, F' \,
  \overline{s}'_{m'}$.
  Therefore, once we have shown that ($\star$) implies ($\dagger$),
  we know that ($\dagger$) holds for $v = u$ and $v' = u'$ and we are done.

  We prove that ($\star$) implies ($\dagger$) by induction on the size of $v$ and $v'$. We consider
  the following cases:

  \begin{description}[leftmargin=3mm]
    \item[$v = \lambda x. \, v_1$]  For ($\star$) to hold, $v$ and $v'$ must be of the
    same type. Therefore, the $\lambda$-prefixes of their $\eta$-long representatives
    must have the same length and we can apply equation
    $\ref{diff:alpha}$. By the induction hypothesis, ($\dagger$) holds.
    \item[$v = x_i$] In this case, $\{x_1 \mapsto s_1, \ldots, x_m \mapsto
    s_m\} v = s_i$. Since ($\star$) holds, $v'$ must be an instance of a unifier from
    the CSU of $s_i = H_i \, \overline{s}'_{m'}$. However, since $s_i$ and all
    terms in $\overline{s}'_{m'}$ are ground, $\lambda \overline{y}_m. \, v' =
    \sigma_i^k(H_i)$, for some $k$. Then, $\textsf{diff}(x_i, v') = z_i^k$, and
    it is easy to check that ($\dagger$) holds.
    \item[$v = x_i \, \overline{v}_n, n > 0$] In this case, $x_i$ is mapped to $s_i$
    which, due to solidity restrictions, has to be a functional bound variable.
    Since ($\star$) holds, we conclude that the head of $\{y_1 \mapsto s'_1, \ldots,
    y_{m'} \mapsto  s'_{m'}\}v'$ must be $s'_{j}$, such that $s'_{j} = s_i$;
    this also means that $v' = y_j \, \overline{v'}_n$. Therefore, it is easy to
    check that some $\tau = \{ H_i \mapsto \lambda \overline{y}_{m'}. \, y_j\}$ is
    a matcher for the problem $s_i = H_i \, \overline{s}'_{m'}$. For some $k$,
    $\sigma_i^k = \tau$, i.e., $\textsf{diff}(v,v') = z_i^k \, \textsf{diff}(v_1, v'_1)
    \ldots \textsf{diff}(v_n, v'_n) $. By induction hypothesis, we get that $(\dagger)$ holds.
    \item[$v = a \, \overline{v}_n$] In the remaining cases $a$ is either a free variable, a loose bound variable
    different than $x_1, \ldots, x_m$, or a constant. If $a$ is a free variable or a loose bound variable
    different than $x_1, \ldots, x_m$, then $v' = a \, \overline{v'}_n$, since
    ($\star$) holds, all of $\overline{s}'_{m'}$ are ground, and bound variables different than
    $x_1, \ldots, x_m$ and $y_1, \ldots, y_{m'}$ are renamed to match by equation~\ref{diff:alpha}.
    By the induction
    hypothesis and by equation~\ref{diff:samehd}, we obtain $(\dagger)$. If $a$ is a constant, we consider two
    cases: either $v' = a \, \overline{v'}_n$, which allows us to apply the induction
    hypothesis and obtain $(\dagger)$ as above, or $v' = y_j \, \overline{v'}_m$.
    Since $a$ is a constant, $y_j$ cannot be a functional bound variable, since
    then it would be mapped to $s'_j$, which due to solidity restrictions
    also has to be a functional bound
    variable and $(\star)$ would not hold. Therefore $v' = y_j$. In this case, we
    proceed as in the case $v = x_i$ with the roles of $v$ and $v'$ swapped. \qedhere
  \end{description}

\end{proof}

\begin{thm}%
    \label{thm:oracle}
    Let $s$ and $t$ be
    solid terms that share no free variables,
    and let $s$ be linear.
    Then the unification problem $\{s \unif t\}$ has a finite CSU\@.
\end{thm}
\begin{proof}
    \looseness=-1
    By Lemma~\ref{lemma:solid-termination-preunif},
    PT terminates on $\{s \unif t\}$ with a finite set of preunifiers $\sigma$,
    each associated with a multiset $E$ of solid flex-flex pairs.

    An MGU $\delta_E$ of the remaining multiset $E$
    of solid flex-flex constraints can be found as follows.
    Choose one constraint $(u \unif v) \in E$ and determine an MGU $\varrho$ for it
    using Lemma~\ref{lemma:solid-var-same-mgu} or~\ref{lemma:solid-vars-diff-mgu}.
    Then the set $\varrho(E\setminus\{u \unif v\})$ also contains only solid flex-flex constraints,
    and we iterate this process by choosing a constraint from $\varrho(E\setminus\{u \unif v\})$ next
    until there are no constraints left, eventually
    yielding an MGU $\varrho'$ of $\varrho(E\setminus\{u \unif v\})$. Finally, we obtain the MGU
    $\delta_E = \varrho'\varrho$ of $E$.

    Let $U = \{ \delta_E\sigma
    \mid \text{PT produces preunifier $\sigma$ with constraints }E\}$.
    By termination of PT, $U$ is finite.
    We show that $U$ is a CSU\@.
    Let $\varrho$ be an arbitrary unifier for $\{ s \unif t\}$. By Lemma~\ref{lemma:pt-complete},
    PT produces a preunifier $\sigma$ with flex-flex constraints $E$
    such that there is a unifier $\theta$ of $E$ and
    $\varrho X = \theta\sigma X$ for all $X$ not contained in the supply of fresh variables $V$.
    Since
    $\delta_E$ is an MGU of $E$,
    assuming that we use variables from $V$ in the role of the auxiliary variables,
    there exists a
    substitution $\theta'$ such that $\theta X = \theta' \delta_E X$ for all $X \not\in V$.
    If we make sure that we never reuse fresh variables and that the supply $V$ does not contain any variables from the initial problem,
    it follows that $\varrho X = \theta' \delta_E\sigma X$ for all $X \not\in V$.
    Therefore, $U$ is a CSU\@.
\end{proof}

The proof of Theorem~\ref{thm:oracle}
provides an effective way to calculate a CSU using PT and the results
of Lemmas~\ref{lemma:solid-var-same-mgu} and~\ref{lemma:solid-vars-diff-mgu}.

 \begin{exa}%
   \label{example:solid}
  Let $\{F \, (\cst{f} \, \cst{a}) \unif \cst{g} \, \cst{a} \, (G \, \cst{a})\}$ be
  the unification problem to solve. Projecting~$F$ onto the first argument will
  lead to a nonunifiable problem, so we perform imitation of~$\cst{g}$ building a
  binding $\sigma_1 = \{ F \mapsto \lambda x.\, \cst{g} \, (F_1 \, x) \, (F_2 \, x)
  \}$. This yields the  problem $\{F_1
  \, (\cst{f} \, \cst{a}) \unif \cst{a}, \allowbreak F_2 \, (\cst{f} \, \cst{a}) \unif G \,
  \cst{a}\}$. Again, we can only imitate $\cst{a}$ for $F_1$---building a new
  binding $\sigma_2 = \{F_1 \mapsto \lambda x.\, \cst{a} \}$. Finally, this
  yields the problem $\{ F_2 \, (\cst{f} \, \cst{a}) \unif G \, \cst{a} \}$.
  According to Lemma~\ref{lemma:solid-vars-diff-mgu}, we find CSUs for the
  problems $J_1 \, \cst{a} = \cst{f} \, \cst{a}$ and $I_1 \, (\cst{f} \, \cst{a})
  \unif \cst{a}$ using PT\@. The latter problem has a singleton CSU $\{I_1 \mapsto \lambda x.\,
  \cst{a}\}$, whereas the former has a CSU containing $\{ J_1 \mapsto \lambda x.\,
  \cst{f} \, x \}$ and $\{ J_1 \mapsto \lambda x.\, \cst{f} \, \cst{a} \}$.
  Combining these solutions, we obtain an MGU $\sigma_3 =
  \{ F_2 \mapsto \lambda x.\, H \, x \, x \, \cst{a},\,
  G   \mapsto \lambda x.\, H \, (\cst{f} \, \cst{a}) \, (\cst{f} \, x) \, x \}$
  for $F_2 \, (\cst{f}\, \cst{a}) \unif G \, \cst{a}$.
  Finally, we get the MGU $\sigma = \sigma_3 \sigma_2 \sigma_1 = \{ F
  \mapsto \lambda x.\, \cst{g} \, \cst{a} \, (H \, x \, x \, \cst{a}),\, G \mapsto
  \lambda x.\, H \, (\cst{f} \, \cst{a}) \, (\cst{f} \, x) \, x  \}$ of the original problem.
  For brevity, we omitted the intermediate bindings of auxiliary variables
  in $\sigma$.
\end{exa}

Example~\ref{example:solid} shows that the solid fragment is useful for
automatic theorem provers based on $\lambda$-superposition~\cite{ab-lamsup-2019}.
When the solid oracle is used, superposing from $F \, (\cst{f} \, \cst{a})$
into $\cst{g} \, \cst{a} \, (G \, \cst{a})$ yields a single clause; without it,
our procedure does not terminate.


Small examples that violate conditions of Theorem~\ref{thm:oracle}
and admit only infinite CSUs can be found easily. The problem $\{ \lambda x.\, F \, (\cst{f} \, x) \unif \lambda x.\, \cst{f} \,
(F \, x) \}$ violates variable distinctness and is a well-known example of a problem with only infinite CSUs.
Similarly,
$\lambda x.\, \cst{g} \, (F \, (\cst{f} \, x)) \, F  \unif \lambda x.\, \cst{g}
\, (\cst{f} \, (G \, x)) \, G $, which violates linearity, reduces to the previous problem.
Only ground arguments to free variables are allowed because  $\{F \, X \unif G
\, \cst{a}\}$ has only infinite CSUs. Finally, it is crucial that
functional arguments to free variables are only bound variables: the problem
$\{ \lambda y. \, X (\lambda x.\, x) \, y \unif \lambda y. \, y \}$ has only infinite CSUs.


\section{An Extension of Fingerprint Indexing}%
\label{sec:indexing}

\looseness=-1
A fundamental building block for almost all automated reasoning tools is the
operation of retrieving term pairs that satisfy certain conditions, e.g.,
unifiable terms, instances or generalizations. Indexing data structures are used
to implement this operation efficiently. If the data structure retrieves
precisely the terms that satisfy the condition, it is called \emph{perfect}.

Higher-order indexing has received little attention compared to its
first-order counterpart. However, recent research in higher-order theorem
proving increased the interest in higher-order indexing~\cite{ls-indexing-16,br-combs-19}.
A \emph{fingerprint index}~\cite{ss-fpindex-12,pv-ehoh-19} is an imperfect index based on
the idea that the skeleton of the term consisting of all non-variable positions
is not affected by substitutions. Therefore, we can easily determine that the terms
are not unifiable (or matchable) if they disagree on a fixed set of sample
positions.

More formally, when we sample an untyped first-order term $t$ on a sample position $p$,
the \emph{generic fingerprinting function} $\text{gfpf}$
distinguishes four possibilities:
\[
  \kern5em
  \text{gfpf}(t,p) =
  \begin{cases}
    \cst{f} & \text{if $t|_p$ has a symbol head $\cst{f}$} \\[-\jot]
    \AA & \text{if $t|_p$ is a variable} \\[-\jot]
    \BB & \text{if $t|_q$ is a variable for some proper prefix $q$ of $p$} \\[-\jot]
    \NN & \text{otherwise}
  \end{cases}
  \kern5em
\]
We define the \emph{fingerprinting function} $\text{fp}(t) = (\text{gfpf}(t,
p_1), \ldots, \text{gfpf}(t, p_n))$, based on a fixed tuple of positions $\overline{p}_n$.
Determining whether two terms are compatible for a given retrieval operation reduces
to checking their fingerprints' componentwise compatibility. The following matrices
determine the compatibility for retrieval operations:
\begin{align*}
    && \hbox{\begin{tabular}{@{}C{1em}|C{1em}C{1em}C{1em}C{1em}C{1em}@{}}
    & $\cst{f}_1$ & $\cst{f}_2$ & \AA & \BB & \NN \\
    \hline
    $\cst{f}_1\vphantom{^{()}}$ & & \xmark & & & \xmark \\
    \AA & & & & & \xmark \\
    \BB & & & & & \\
    \NN & \xmark & \xmark & \xmark & &
    \end{tabular}}
    &&
    \hbox{\begin{tabular}{@{}C{1em}|C{1em}C{1em}C{1em}C{1em}C{1em}@{}}
    & $\cst{f}_1$ & $\cst{f}_2$ & \AA & \BB & \NN \\
    \hline
    $\cst{f}_1\vphantom{^{()}}$ & & \xmark & \xmark & \xmark & \xmark \\
    \AA & & & & \xmark & \xmark \\
    \BB & & & & & \\
    \NN & \xmark & \xmark & \xmark & \xmark &
    \end{tabular}}
\end{align*}
The left matrix determines unification compatibility, while the right matrix
determines compatibility for matching term $s$ (rows) onto term $t$ (columns). Symbols
$\cst{f}_1$ and $\cst{f}_2$ stand for arbitrary distinct constants. Incompatible
features are marked with \xmark. For example, given a tuple of term positions $(1, 1.1.1, 2)$, and
terms $\cst{f}(\cst{g}(X), \cst{b})$ and $\cst{f}(\cst{f}(\cst{a},\cst{a}),\cst{b})$,
their fingerprints are $(\cst{g}, \BB, \cst{b})$ and $(\cst{f},
\NN, \cst{b})$, respectively. Since the first fingerprint component is incompatible,
terms are not unifiable.

Fingerprints for the terms in the index are stored in a trie data structure.
This allows us to efficiently filter out terms that are not compatible with a given
retrieval condition. For the remaining terms, a unification or matching
procedure must be invoked to determine whether they satisfy the condition or not.

The fundamental idea of first-order fingerprint indexing carries over to
higher-order terms---application of a substitution does not change the rigid skeleton of a term.
However, to extend fingerprint indexing to
higher-order terms, we must address the issues of $\alpha\beta\eta$-normalization and figure how to cope with $\lambda$-abstractions and
bound variables. To that end, we define a function $\hotofo{t}$, defined on
$\beta$-reduced $\eta$-long terms in De Bruijn~\cite{db-dbindices-75} notation:
\begin{align*}
& \hotofo{F \, \overline{s}} = F
&& \hotofo{x_i \, \overline{s}_n } = \cstdba{i}(\hotofo{s_1}, \ldots, \hotofo{s_n})
&& \hotofo{\cst{f} \, \overline{s}_n } = \cst{f}(\hotofo{s_1}, \ldots, \hotofo{s_n}) &&
\hotofo{\lambda \overline{x}. \, s  } = \hotofo{s}
\end{align*}
\looseness=-1
We let $x_i$ be a bound variable of type $\alpha$ with De Bruijn index $i$,
and $\cstdba{i}$ be a fresh constant corresponding to this variable. All constants
$\cstdba{i}$ must be fresh.
Effectively, $\hotofo{\;}$ transforms a $\eta$-long $\beta$-reduced higher-order term to an untyped first-order term. Let $\nf{t}{\beta\eta}$ be
the $\eta$-long $\beta$-reduced form of $t$; the higher-order generic
fingerprinting function $\text{gfpf}_\textsf{ho}$, which relies on conversion $\nametodb{t}$ from
named to De Bruijn representation, is defined as
\[
  \text{gfpf}_\textsf{ho}(t,p) =
    \text{gfpf}(\hotofo{\nametodb{\nf{t}{\beta\eta}}}, p)
\]

If we define $\text{fp}_\textsf{ho}(t) =
\text{fp}(\hotofo{\nametodb{\nf{t}{\beta\eta}}})$, we can support fingerprint indexing
for higher-order terms with no changes to the compatibility matrices. For example,
consider the terms $s = (\lambda x y. \, x \, y) \, \textsf{g}$ and $t = \cst{f}$,
where $\cst{g}$ has the type $\alpha \rightarrow \beta$ and $\cst{f}$ has the
type $\alpha \rightarrow \alpha \rightarrow \beta$.
For the tuple of
positions $(1, 1.1.1, 2)$ we get
\begin{align*}
  & \text{fp}_\textsf{ho}(s) =
    \text{fp}(\hotofo{\nametodb{\nf{s}{\beta\eta}}}) =
    \text{fp}(\cst{g}(\cstdba{0})) =
    (\cstdba{0}, \NN, \NN) \\
  & \text{fp}_\textsf{ho}(t) =
   \text{fp}(\hotofo{\nametodb{\nf{t}{\beta\eta}}}) =
    \text{fp}(\cst{f}(\cstdba{1}, \cstdba{0})) =
    (\cstdba{1}, \NN, \cstdba{0})
\end{align*}
Since the first and third fingerprint component are incompatible, the terms are not unifiable.

\looseness=-1
Other first-order indexing techniques such as feature vector indexing and substitution trees
can probably be extended to higher-order terms using the method described here as well.

\section{Implementation}%
\label{sec:implementation}
Zipperposition~\cite{sc-phd-2015,sc-supind-17} is an
open-source\footnote{\url{https://github.com/sneeuwballen/zipperposition}}
theorem prover written in OCaml.
It is a versatile testbed for prototyping extensions to
superposition-based theorem provers.
It was initially
designed as a prover for polymorphic first-order logic
and then extended to higher-order logic.
A recent addition is a complete mode for
Boolean-free higher-order logic~\cite{ab-lamsup-2019},
which depends on a unification procedure that can enumerate a CSU\@.
We implemented our procedure in Zipperposition.

We used OCaml's functors to create a modular implementation. The core of our
procedure is implemented in a module which is parametrized by another module
providing oracles and implementing the \textsf{Bind} step. In this way we can
obtain the complete or pragmatic procedure and seamlessly integrate
oracles while reusing as much common code as possible.

To enumerate all elements of a possibly infinite CSU, we rely on lazy lists whose
elements are subsingletons of unifiers (either one-element sets containing a unifier
or empty sets). The search space must be explored in a \emph{fair} manner,
meaning that no branch of the constructed tree is indefinitely postponed.

Each {\sf Bind} step will give rise to new unification problems $E_1, E_2,
\ldots$ to be solved. Solutions to each of those problems are lazy lists $p_1,
p_2, \ldots$ containing subsingletons of unifiers. To avoid postponing some
unifier indefinitely, we use the dovetailing technique: we first take one
subsingleton from $p_1$, then one from each of $p_1$ and $p_2$. We continue with
one subsingleton from each of $p_1,p_2$ and $p_3$, and so on. Empty lazy lists are
ignored in the traversal. To ensure we do not remain stuck waiting for a unifier
from a particular lazy list, the procedure will periodically return an empty
set, indicating that the next lazy list should be probed.

\looseness=-1
The implemented selection function for our procedure prioritizes selection of
rigid-rigid over flex-rigid pairs, and flex-rigid over flex-flex pairs. However,
since the constructed substitution $\sigma$ is not applied eagerly, heads can
appear to be flex, even if they become rigid after dereferencing and
normalization. To mitigate this issue to some degree, we
dereference the heads with $\sigma$, but do not normalize, and use the resulting
heads for prioritization.

We implemented oracles for the pattern, solid, and fixpoint fragment.
Fixpoint unification~\cite{gh-unification-75} is concerned with problems of the form
$\{F \unif t\}$. If $F$ does not occur in $t$, $\{ F \mapsto t \}$
is an MGU for the problem. If there is a position $p$ in $t$ such that
$t|_p = F \, \overline{u}_m$
and for each prefix $q\neq p$ of $p$, $t|_q$ has a rigid head and either $m=0$ or $t$ is not a $\lambda$-abstraction,
then we can conclude that $F \unif t$ has no solutions. Otherwise,
the fixpoint oracle is not applicable.

For second-order logic with only unary constants, it is decidable whether
a unifier for a problem in this class (called
\emph{monadic second-order}) exists~\cite{wf-monadicunif-88}. As this class of terms admits a
possibly infinite CSU, this oracle cannot be used for \textsf{OracleSucc} step,
but it can be used for \textsf{OracleFail}. Similarly the fragment of
second-order terms with no repeated occurrences of free variables has decidable
unifier existence but possibly infinite CSUs~\cite{gd-unif-chapter-01}. Due to their limited applicability and high complexity we
decided not to implement these oracles.




\section{Evaluation}%
\label{sec:evaluation}

We evaluated the implementation of our unification procedure in Zipperposition,
assessing the complete variant and the pragmatic variant, the latter with several
different combinations of limits for number of bindings. As part of the
implementation of the complete mode for Boolean-free higher-order logic in
Zipperposition~\cite{ab-lamsup-2019}, Bentkamp implemented a straightforward
version of the JP procedure. This version is faithful to the original description,
with a check as to whether a (sub)problem can be solved using a first-order oracle 
as the only optimization. Our evaluations were performed on StarExec Miami~\cite{as-starexec-14} servers with Intel Xeon E5-{2620 v4} CPUs clocked at {2.10}\,GHz with 60\,s CPU limit.

Contrary to first-order unification, there is no widely available corpus of
benchmarks dedicated solely to evaluating performance of higher-order
unification algorithms. Thus, we used all 2606 monomorphic higher-order theorems
from the TPTP library~\cite{gf-tptp-17} and 832 monomorphic higher-order
Sledgehammer (SH) generated problems~\cite{ns-sh-13} as our benchmarks\footnote{An archive with raw
results, all used problems, and scripts for running each configuration is
available at \url{http://doi.org/10.5281/zenodo.4269591}}.
Many TPTP problems require synthesis of complicated unifiers,
whereas Sledgehammer problems are only mildly higher-order---many of them are
solved with first-order unifiers.

\looseness=-1
We used the naive implementation of the JP procedure (\textbf{jp}) as a baseline
to evaluate the performance of our procedure. We compare it with the complete
variant of our procedure (\textbf{cv}) and pragmatic variants
(\textbf{pv}) with several different configurations of limits for
applied bindings. All other Zipperposition parameters have been fixed to the values of
a variant of a well-performing configuration we used for the CASC-27 theorem proving competition~\cite{gs-casc-2019}.
The cv configuration and all of the pv
configurations use only pattern unification as an underlying oracle. To test
the effect of oracle choice, we evaluated the complete variant
in 8 combinations: with no oracles
(\textbf{n}), with only fixpoint (\textbf{f}), pattern (\textbf{p}), or solid (\textbf{s}) oracle,
and with their combinations: \textbf{fp}, \textbf{fs}, \textbf{ps}, \textbf{fps}.

\looseness=-1
Figure~\ref{fig:new-old} compares different variants of the procedure with the
naive JP implementation. Each pv configuration is denoted by pv$^{a}_{bcde}$
where $a$ is the limit on the total number of applied bindings, and $b$, $c$,
$d$, and $e$ are the limits of functional projections, eliminations, imitations,
and identifications, respectively. Figure~\ref{fig:oracles} summarizes the
effects of using different oracles.

The configuration of our procedure with no oracles outperforms the JP procedure with the first-order oracle. This
suggests that the design of the procedure, in particular lazy normalization and
lazy application of the substitution, already reduces the effects of the JP
procedure's main bottlenecks. Raw evaluation data shows that on TPTP benchmarks, complete and pragmatic configurations
differ in the set of problems they solve---cv solves 19
problems not solved by pv$^{4}_{2222}$, whereas pv$^{4}_{2222}$ solves 34
problems cv does not solve. Similarly, comparing the pragmatic configurations with each other,
pv$^{6}_{3333}$ and pv$^{4}_{2222}$ each solve 13 problems that the other one does not.  The overall higher success rate of pv$^2_{1020}$ compared to
pv$^2_{1222}$ suggests that solving flex-flex pairs by trivial unifiers often
suffices for superposition-based theorem proving.

\newbox\mybox
\setbox\mybox=\hbox{\small pv$^{12}_{66666}$}

\newcommand\HEAD[1]{\hbox to \wd\mybox{\hfill\hbox{#1}\hfill}}
\newcommand\Z{\phantom{0}}
\newcommand\MIDLINE{\\[.25ex]\hline\rule{0pt}{3ex}}

\begin{figure}[t]
	\def\arraystretch{1.1}%
  \begin{tabular}{lcccccccc}
    \strut                  &  \HEAD{jp}               & \HEAD{cv}              & \HEAD{pv$^{12}_{6666}$}    & \HEAD{pv$^{6}_{3333}$}     & \HEAD{pv$^{4}_{2222}$}    & \HEAD{pv$^{2}_{1222}$}   & \HEAD{pv$^{2}_{1121}$}         & \HEAD{pv$^{2}_{1020}$}
    \MIDLINE
    TPTP                    &  \eqmakebox[jp][r]{1551} &  \eqmakebox[nc][r]{1717} &  \eqmakebox[np6][r]{1722} &  \eqmakebox[np3][r]{\bf 1732} & \eqmakebox[np2][r]{\bf 1732} & \eqmakebox[np1][r]{1715} &  \eqmakebox[np10][r]{1712}   & \eqmakebox[np0][r]{1719}  \\
    SH                      &  \eqmakebox[jp][r]{242}  &  \eqmakebox[nc][r]{\bf 260}  &  \eqmakebox[np6][r]{253}  &  \eqmakebox[np3][r]{255}  & \eqmakebox[np2][r]{255}  & \eqmakebox[np1][r]{254}  &  \eqmakebox[np10][r]{259}    & \eqmakebox[np0][r]{257}
    \end{tabular}
	\caption{Proved problems, per configuration}%
	\label{fig:new-old}
\end{figure}
\begin{figure}[t]
	\def\arraystretch{1.1}%
  \begin{tabular}{lcccccccc}
    \strut                  &  \HEAD{n}                     & \HEAD{f}                      & \HEAD{p}                       & \HEAD{s}                       & \HEAD{fp}                     & \HEAD{fs}                     & \HEAD{ps}                      & \HEAD{fps}
    \MIDLINE
    TPTP                    &  \eqmakebox[n][r]{1658} & \eqmakebox[s][r]{1717} &  \eqmakebox[p][r]{1717} & \eqmakebox[sp][r]{1720} & \eqmakebox[n][r]{1719} & \eqmakebox[s][r]{\bf 1724} &  \eqmakebox[p][r]{1720} & \eqmakebox[sp][r]{1723}  \\
    SH                      &  \eqmakebox[n][r]{245}  & \eqmakebox[s][r]{255}  &  \eqmakebox[p][r]{\bf 260}  & \eqmakebox[sp][r]{259}  & \eqmakebox[n][r]{255}  & \eqmakebox[s][r]{254}  &  \eqmakebox[p][r]{258}  & \eqmakebox[sp][r]{254}
    \end{tabular}
	\caption{Proved problems, per used oracle}%
	\label{fig:oracles}
\end{figure}

Counterintuitively, in some cases, using oracles can hurt the performance of
Zipperposition. Using oracles typically results in generating smaller
CSUs, whose elements are more general substitutions than the ones we obtain without oracles.
These more general substitutions usually
contain more applied variables, which Zipperposition's heuristics avoid due to
their explosive nature. This can make Zipperposition postpone necessary inferences for too long.
Configuration n benefits from this effect and therefore solves 18 TPTP problems that no other configuration in Figure~\ref{fig:oracles} solves. The same effect also gives configurations with only one oracle
an advantage over configurations with multiple oracles on some problems.

The evaluation sheds some light on how often solid unification problems appear in practice.
The raw data show that
configuration s solves 5 TPTP problems that neither f nor p solve. Configuration
f solves 8 TPTP problems that neither s nor p solve, while p solves 9 TPTP
problems that two other configurations do not. This suggests that the solid oracle
is slightly less beneficial than the fixpoint or pattern oracles, but still presents a
useful addition the set of available oracles.

\looseness=-1
A subset of TPTP benchmarks, concerning operations on Church numerals, is designed to test the efficiency of
higher-order unification. Our procedure performs exceptionally well
on these problems---it solves all of them, usually
faster than other competitive higher-order provers.
  There are 11 benchmarks in \texttt{NUM} category of TPTP that contain conjectures about
  Church numerals: \texttt{NUM020\^{}1}, \texttt{NUM021\^{}1}, \texttt{NUM415\^{}1},
  \texttt{NUM416\^{}1}, \texttt{NUM417\^{}1}, \texttt{NUM418\^{}1}, \texttt{NUM419\^{}1}, \texttt{NUM798\^{}1},
  \texttt{NUM799\^{}1}, \texttt{NUM800\^{}1}, and \texttt{NUM801\^{}1}.
  We evaluated those problems using the same CPU nodes and the same time limits
  as above. In addition to Zipperposition, we used
  all higher-order provers that took part in the 2019 edition of CASC~\cite{gs-casc-2019} (in the THF category) for this evaluation: CVC4
  1.7~\cite{cb-cvc4-11}, Leo-III 1.4~\cite{ascb-leo3-2018}, Satallax 3.4~\cite{cb-satallax-12}, Vampire 4.4 THF~\cite{lc-vampire-13}. Figure~\ref{fig:numerals}
  shows the CPU time needed to solve a problem or ``--'' if the prover timed out.

  \begin{figure}[tb]
    \def\arraystretch{1.1}%
    \centerline{\begin{tabular}{@{}l@{\hskip 1.25em}c@{\hskip 1em}c@{\hskip 1em}c@{\hskip 1em}c@{\hskip 1em}c@{\hskip 1em}c}
      \strut                  &  CVC4               & Leo-III                & Satallax             & Vampire               & Zipperposition (cv)
      \MIDLINE
      \texttt{NUM020\^{}1}  &  \eqmakebox[n][r]{--} & \eqmakebox[n][r]{0.46} & \eqmakebox[n][r]{--} & \eqmakebox[n][r]{--} & \eqmakebox[n][r]{\bf 0.03} \\
      \texttt{NUM021\^{}1}  &  \eqmakebox[n][r]{--} & \eqmakebox[n][r]{--} & \eqmakebox[n][r]{--} & \eqmakebox[n][r]{--} & \eqmakebox[n][r]{\bf 4.10} \\
      \texttt{NUM415\^{}1}  &  \eqmakebox[n][r]{45.80} & \eqmakebox[n][r]{0.34} & \eqmakebox[n][r]{0.21} & \eqmakebox[n][r]{0.42} & \eqmakebox[n][r]{\bf 0.03} \\
      \texttt{NUM416\^{}1}  &  \eqmakebox[n][r]{47.37} & \eqmakebox[n][r]{0.92} & \eqmakebox[n][r]{0.21} & \eqmakebox[n][r]{0.41} & \eqmakebox[n][r]{\bf 0.07} \\
      \texttt{NUM417\^{}1}  &  \eqmakebox[n][r]{--} & \eqmakebox[n][r]{49.73} & \eqmakebox[n][r]{\bf 0.30} & \eqmakebox[n][r]{0.40} & \eqmakebox[n][r]{0.45} \\
      \texttt{NUM418\^{}1}  &  \eqmakebox[n][r]{--} & \eqmakebox[n][r]{0.40} & \eqmakebox[n][r]{1.29} & \eqmakebox[n][r]{0.38} & \eqmakebox[n][r]{\bf 0.03} \\
      \texttt{NUM419\^{}1}  &  \eqmakebox[n][r]{--} & \eqmakebox[n][r]{0.42} & \eqmakebox[n][r]{23.33} & \eqmakebox[n][r]{0.37} & \eqmakebox[n][r]{\bf 0.03} \\
      \texttt{NUM798\^{}1}  &  \eqmakebox[n][r]{46.29} & \eqmakebox[n][r]{0.35} & \eqmakebox[n][r]{4.01} & \eqmakebox[n][r]{0.38} & \eqmakebox[n][r]{\bf 0.03} \\
      \texttt{NUM799\^{}1}  &  \eqmakebox[n][r]{--} & \eqmakebox[n][r]{5.05} & \eqmakebox[n][r]{--} & \eqmakebox[n][r]{--} & \eqmakebox[n][r]{\bf 0.03} \\
      \texttt{NUM800\^{}1}  &  \eqmakebox[n][r]{--} & \eqmakebox[n][r]{--} & \eqmakebox[n][r]{--} & \eqmakebox[n][r]{\bf 0.37} & \eqmakebox[n][r]{3.15} \\
      \texttt{NUM801\^{}1}  &  \eqmakebox[n][r]{--} & \eqmakebox[n][r]{0.73} & \eqmakebox[n][r]{38.77} & \eqmakebox[n][r]{--} & \eqmakebox[n][r]{\bf 0.50} \\
      \end{tabular}}
    \caption{Time needed to prove a problem, in seconds.}%
    \label{fig:numerals}
  \end{figure}

\section{Discussion and Related Work}
\looseness=-1
The problem addressed in this paper is that of finding
a complete and efficient higher-order unification procedure. Three main lines of
research dominated the research field of higher-order unification over the last forty years.

\looseness=-1
The first line of research went in the direction of finding procedures that
enumerate CSUs. The most prominent procedure designed for
this purpose is the JP procedure~\cite{jp-unif-76}. Snyder and Gallier~\cite{sg-unif-89} also provide such a procedure,
but instead of solving
flex-flex pairs systematically, their procedure blindly guesses the head of the necessary
binding by considering all constants in the signature and fresh variables of all
possible types. Another approach, based on higher-order combinators, is given by
Dougherty~\cite{dd-combunif-93}. This approach blindly creates (partially
applied) \textsf{S}-, \textsf{K}-, and \textsf{I}-combinator bindings for applied
variables, which results in returning many redundant unifiers, as well as in
nonterminating behavior even for simple problems such as $X \, \cst{a} =
\cst{a}$.

The second line of research is concerned with enumerating
preunifiers. The most prominent procedure in this line of research is Huet's~\cite{gh-unification-75}. The Snyder--Gallier procedure restricted not to solve flex-flex pairs 
is a version of the PT procedure presented in Section~\ref{sec:solid-oracle}. It improves
Huet's procedure by featuring a Solution rule.

The third line of research gives up the expressiveness of the full
$\lambda$-calculus and focuses on decidable fragments. Patterns~\cite{tn-patterns-93} are arguably the most important such fragment in practice,
with implementations
in Isabelle~\cite{tn-isabelle-2002}, Leo-III~\cite{ascb-leo3-2018}, Satallax~\cite{cb-satallax-12}, $\lambda$Prolog~\cite{dm-lprolog-12}, and other systems.
Functions-as-constructors~\cite{tl-facunif-2016} unification subsumes pattern
unification but is significantly more complex to implement. Prehofer~\cite{cp-unifphd-95} lists many other decidable fragments, not only for
unification but also preunification and unifier existence problems. Most of
these algorithms are given for second-order terms with various constraints on
their variables. Finally, one of the first decidability results is Farmer's discovery~\cite{wf-monadicunif-88} that
higher-order unification of terms with unary function symbols is decidable.

Our procedure draws inspiration from and contributes to all three lines of research.
Accordingly, its advantages over previously known procedures can be laid out
along those three lines. First, our procedure mitigates many issues of the
JP procedure. Second, it can be modified not to solve flex-flex pairs,
and become a version of Huet's procedure with important built-in optimizations.
Third, it can integrate any oracle for problems with finite CSUs, including the one we discovered.


The implementation of our procedure in Zipperposition was one of the reasons
this prover evolved from proof-of-concept prover for higher-order logic to
competitive higher-order prover. In the 2020 edition of CASC, Zipperposition won
the higher-order division solving 84\% of problems, which is 20 percentage
points ahead of the runner-up.

\section{Conclusion}
\looseness=-1
We presented a procedure for enumerating a complete set of higher-order unifiers
that is designed for efficiency. Due to a design that restricts
the search space and a tight integration of oracles,
 it reduces the number of redundant unifiers returned and
gives up early in cases of nonunifiability. In addition, we presented a new
fragment of higher-order terms that admits finite CSUs. Our evaluation shows
a clear improvement over previously known procedures.

In future work, we will focus on designing intelligent heuristics that
automatically adjust unification parameters according to the type of the problem.
For example, we should usually choose shallow unification for mostly first-order
problems and deeper unification for hard higher-order problems. We plan to
investigate other heuristic choices, such as the order of bindings and the way in which search
space is traversed (breadth- or depth-first).
We are also interested in further improving the
termination behavior of the procedure, without sacrificing completeness. Finally,
following the work of Libal~\cite{tl-regularpatterns-15} and Zaionc~\cite{mz-regularunif-85}, we would like to consider the use of
regular grammars to finitely present infinite CSUs. For example, the grammar
$G ::= \lambda x.\, x \mid \lambda x.\, \cst{f} \, (G \, x)$
represents
all elements of the CSU for the problem $\lambda x.\, G \, (\cst{f} \, x) \unif
\lambda x.\, \cst{f} \, (G \, x)$.

\subparagraph{\bf Acknowledgement}
\looseness=-1
We are grateful to the maintainers of StarExec for letting us use their
service.
We thank Ahmed Bhayat, Jasmin Blanchette, Daniel El Ouraoui,
Mathias Fleury, Pascal Fontaine, Predrag Jani\v{c}i\'c,
Robert Lewis, Femke van Raamsdonk, Hans-J\"org Schurr, Sophie Tourret, Dmitriy Traytel,
and the anonymous reviewers for suggesting many improvements to this text.
Vukmirovi\'c and Bentkamp's research has received funding from the European
Research Council (ERC) under the European Union's Horizon 2020 research and
innovation program (grant agreement No.\ 713999, Matryoshka).
Nummelin has received funding from the Netherlands Organization for
Scientific Research (NWO) under the Vidi program (project No.\
016.Vidi.189.037, Lean Forward).

\appendix
\bibliographystyle{alpha}
\bibliography{bib}

\end{document}